\documentclass[sigconf]{acmart}

\usepackage{booktabs} 
\usepackage{balance}
\usepackage{graphicx}
\usepackage{algorithm}
\usepackage{algpseudocode}
\usepackage{xcolor}
\usepackage{ifthen}

\newboolean{shortversion}
\setboolean{shortversion}{false}

\newtheorem{theorem}{Theorem}
\newtheorem{definition}{Definition}

\newcommand{\Paragraph}[1]{\smallskip\noindent{\bf #1}}

\newcommand{\allref}{aldridge14,atia12,barg17,bui13,bui17,cai17,chan14,cheraghchi09,cheraghchi13,cheraghchi11,gilbert12,gilbert08,inan17,indyk10,lee16,mazumdar12,mazumdar14,mazumdar16,ngo11,porat11,sebo85,vem17,zhigljavsky03}

\begin{document}

\title{Cross-Sender Bit-Mixing Coding}
\titlenote{The authors of this paper are alphabetically ordered.}

\author{Steffen Bondorf}
\authornote{Work was done while this author was in National University of Singapore.}
\affiliation{NTNU Trondheim, Norway}
\email{steffen.bondorf@ntnu.no}

\author{Binbin Chen}
\affiliation{Advanced Digital Sciences Center}
\email{binbin.chen@adsc-create.edu.sg}

\author{Jonathan Scarlett}
\affiliation{National University of Singapore}
\email{scarlett@comp.nus.edu.sg}

\author{Haifeng Yu}
\affiliation{National University of Singapore}
\email{haifeng@comp.nus.edu.sg}

\author{Yuda Zhao}
\authornote{Work was done while this author was in National University of Singapore.}
\affiliation{Advance.AI}
\email{yudazhao@gmail.com}

\begin{abstract}
{\em
Scheduling to avoid packet collisions is a long-standing challenge in networking, and has become even trickier in wireless networks with multiple senders and multiple receivers. In fact, researchers have proved that even {\em perfect} scheduling can only achieve $\mathbf{R} = O(\frac{1}{\ln N})$. Here $N$ is the number of nodes in the network, and $\mathbf{R}$ is the {\em medium utilization rate}.

Ideally, one would hope to achieve $\mathbf{R} = \Theta(1)$, while avoiding all the complexities in scheduling. To this end, this paper proposes {\em cross-sender bit-mixing coding} ({\em BMC}), which does not rely on scheduling. Instead, users transmit simultaneously on suitably-chosen slots, and the amount of overlap in different user's slots is controlled via coding.
We prove that in all possible network topologies, using BMC enables us to achieve $\mathbf{R}=\Theta(1)$. We also prove that the space and time complexities of BMC encoding/decoding are all low-order polynomials.
}
\end{abstract}



\maketitle

\begin{figure}
  \centering
  \hspace*{4mm}\includegraphics[scale=.45]{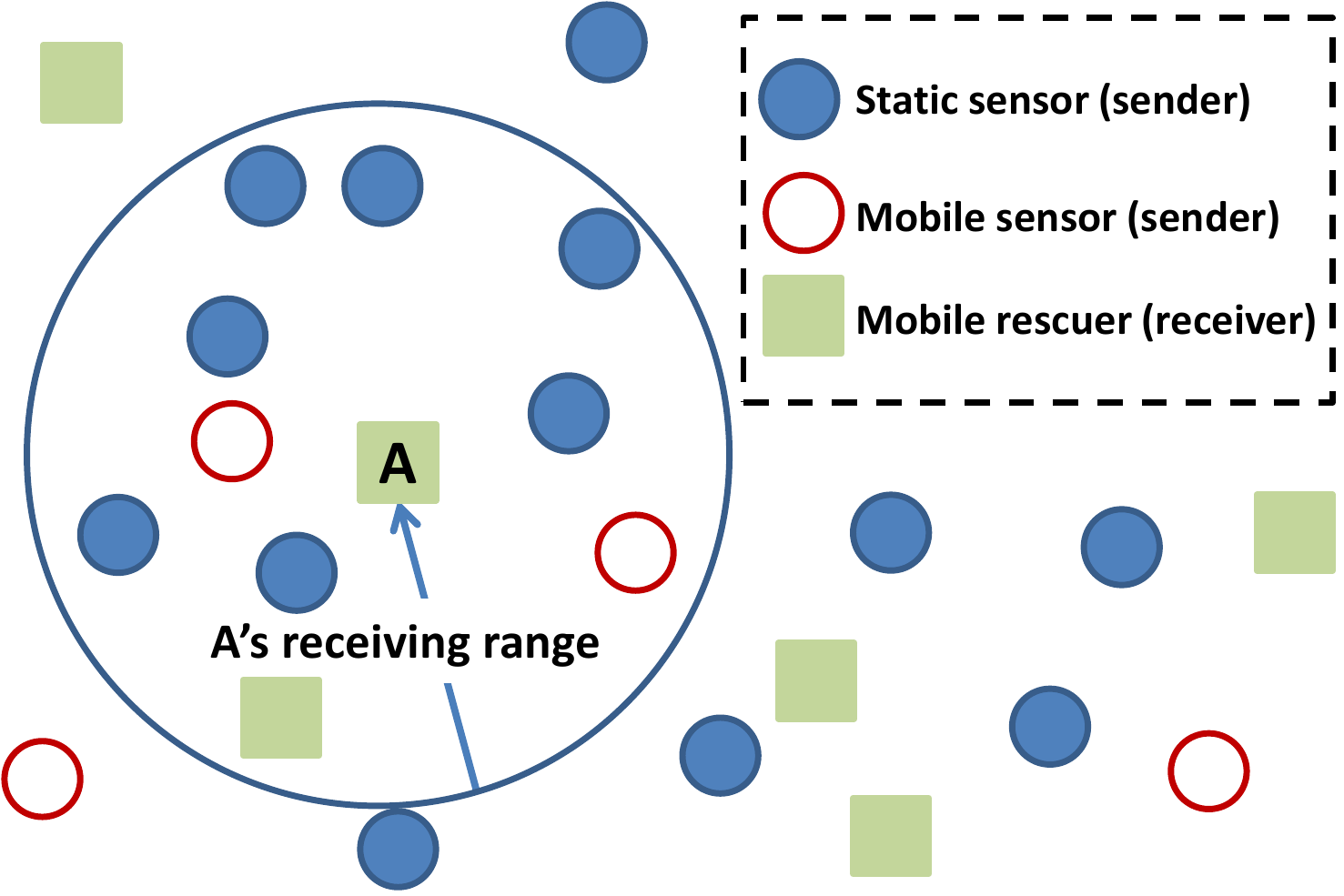}
  \vspace{-0mm}
  \caption{A disaster recovery scenario.}
  \vspace{-5mm}
  \label{fig:app}
\end{figure}

\section{Introduction}
\label{sec:intro}

\Paragraph{Background and motivation.}
Wireless networking relies on a shared communication medium. To avoid packet collision in such a shared medium, a central theme of wireless networking research, since the very beginning, has been to properly schedule/coordinate the senders of the packets. Such scheduling turns out to be complicated, involving a number of challenges such as the hidden terminal problem~\cite{tobagi1975packet}, the exposed terminal problem~\cite{vutukuru2008harnessing}, ACK implosion problem~\cite{zheng08book}, fairness issues among the senders, as well as the lack of global information regarding the network topology.

The growth of wireless networking over the past decade, unfortunately, has made this old problem trickier. Many wireless networks today (or in the near future) are {\em multi-sender multi-receiver networks}, as in the following examples:
\begin{itemize}
\item Consider a disaster recovery scenario (e.g., forest fire or earthquake), where wireless sensors have already been deployed in the environment prior to the disaster. There may also be additional mobile sensors, such as drones, deployed during disaster recovery. There are a number of human rescuers, at different locations. Each rescuer needs to collect information from all the sensors in his/her neighborhood (Figure~\ref{fig:app}). A sensor may belong to the neighborhood of multiple rescuers, and hence needs to send information to all of them.
\item In recent years, vehicular ad-hoc networks (VANET) \linebreak \cite{olariu09book} have been moving closer to reality. In a VANET, each vehicle is simultaneously a sender and a receiver of information. Road safety applications of VANET often require each vehicle to collect information from all its neighboring vehicles. This again results in a multi-sender multi-receiver scenario similar to the above example.
\item The past few years have witnessed the deployment of distributed energy resources (such as solar panels, batteries, and electric vehicles) in power grid systems. Each energy resource can be controlled by an intelligence device, and all these devices can collectively form a wireless network to carry out peer-to-peer collaboration (such as energy trading, demand response, and grid stability control). This results in a multi-sender multi-receiver scenario that could involve hundreds or even thousands of participating devices.
\end{itemize}

%

\Paragraph{Fundamental inefficiency.}
While scheduling in wireless networks is already complicated, in multi-sender multi-receiver networks, scheduling hits a new barrier which is unfortunately {\em fundamental}. To understand, let us consider the example scenario in Figure~\ref{fig:app}. Assume that each rescuer needs to receive a single $d$-byte (e.g., $d=100$) packet
from each of his/her neighboring sensors (i.e., sensors within the rescuer's communication range). For instance, the packet may contain some {\em data item} representing the sensor reading and other information. Also assume that each rescuer has at most $k$ (e.g., $k=100$) neighboring sensors.
Now since each receiver only needs to receive at most $kd$ bytes of information, one might reasonably hope all receivers to receive their respective $O(kd)$ bytes of information within $O(kd)$ time, assuming proper scheduling.

But unfortunately, Ghaffari et 
al.~\cite{ghaffari12}
have proved a strong impossibility result: In the above scenario, even with {\em perfect} scheduling (which assumes perfect global and future knowledge, perfect coordination, as well as infinite computational power), it takes $\Omega(kd\ln N)$ time for all the receivers to receive the packets from their respective senders,
under certain topologies.\footnote{Actually this holds in almost
100\% (or more precisely, $1-\frac{2}{N^2}$ fraction) of their randomly constructed topologies. Furthermore, their proof and lower bound continue to hold even if we allow $O(\frac{1}{N})$ probability of delivery failure for each receiver.} Here $N$ is the total number of nodes in the network (including both senders and receivers), which can be much larger than $k$.

Fundamentally, the multiplicative $\ln N$ term in their lower bound is due to the fact that the best schedules (for sending the packets) with respect to different receivers are incompatible with each other. Hence even though for each receiver there exists a good schedule of $O(kd)$ length, there is no way to merge these schedules into a globally good schedule of $O(kd)$ length.

One should further keep in mind that since their lower bound assumes perfect scheduling, the actual performance of scheduling in practice will likely be much worse.
This lower bound of $\Omega(kd\ln N)$~\cite{ghaffari12} reveals the {\em fundamental} inefficiency of scheduling, in multi-sender multi-receiver wireless networks. It implies that in the above scenario, even with perfect scheduling, the {\em medium utilization rate} (denoted as $\mathbf{R}$) of the wireless network will be at most:
\begin{eqnarray*}
\mathbf{R} &=&
\frac{\mbox{ max \# of useful bits received by a receiver}}{\mbox{\# of bits of airtime used}} \\
&=& \frac{O(kd)}{\Omega(kd\ln N)} = O\Big(\frac{1}{\ln N}\Big)
\end{eqnarray*}
This is undesirable since $\mathbf{R}\rightarrow 0$ as the system size ($N$) increases. Putting it another way, we cannot even utilize any small {\em constant fraction} of the wireless medium, if we rely on scheduling.

\Paragraph{The ultimate goal.}
Ideally, one would hope to achieve a constant utilization rate of the wireless medium in the above setting, and also to greatly simplify the design by
avoiding scheduling altogether.
Namely, we hope to take $O(kd)$ time (which is asymptotically optimal) for all the receivers to receive their respective $O(kd)$ bytes of information. Doing so will overcome the lower bound of $\Omega(kd\ln N)$~\cite{ghaffari12} on scheduling, and improve $\mathbf{R}$ from $O(\frac{1}{\ln N})$ to $\Theta(1)$.

\Paragraph{Our results.}
As a key step to achieving the above ultimate goal, this paper proposes {\em cross-sender bit-mixing coding} (or {\em BMC} in short), as the {\em theoretical underpinning}.
If we use BMC in the previous example scenario, then each sensor will simply encode its data item using BMC, and then send the encoding result, {\em without doing any scheduling and simultaneously with all other sensors}. The packets will be superimposed onto each other, and with BMC decoding, a receiver will recover the original data items.

The main technical developments in this paper center around the design and formal analysis of BMC. We will prove
that in {\em all} possible network topologies and under reasonable parameter ranges (specifically, as long as $k=  \omega(\ln N)$ and $d = \omega(\ln^2 N\times \ln\ln N)$), using BMC in the earlier scenario enables the completion of the transmissions of all the data items in optimal $\Theta(kd)$ time, and hence achieves $\mathbf{R}=\Theta(1)$.
In terms of the overheads, we will prove
that the space and time complexities of our BMC encoding/decoding algorithm are all low-order polynomials, allowing efficient implementations.
Finally, to supplement our formal results, we also provide some basic numerical examples on BMC's benefits and complexity.

We hope that our theoretical results in this work can attest the promise of this direction, and spur future systems research (especially on the physical layer) along this line.

\Paragraph{Superimposed code.}
Our overall approach is reminiscent of the decades-old idea of superimposed code.
In fact, BMC can be viewed as a kind of superimposed code~\cite{kautz64}, which in turn is (almost) equivalent to non-adaptive group testing (NAGT)~\cite{du99book}. There have been numerous designs~\cite{\allref} for superimposed code and NAGT. But applying these existing designs to our context will not enables us to improve $\mathbf{R}$ from $O(\frac{1}{\ln N})$ to $\Theta(1)$:
Many of these designs would incur an
{\em exponential} computational complexity of $\Omega(2^{8d})$ in our context, rendering them infeasible. None of the remaining designs (with polynomial computational complexity) can achieve $\mathbf{R}=\Theta(1)$ in our context (see Section~\ref{sec:related}). To our knowledge, BMC is the very first superimposed code that can achieve $\mathbf{R}=\Theta(1)$ without incurring exponential complexity.



\Paragraph{Roadmap.}
Section~\ref{sec:related} discusses related work.
Section~\ref{sec:overview} gives an overview of our BMC design, while
Section~\ref{sec:assumption} discusses BMC's assumptions on the physical layer.
Section~\ref{sec:details} and \ref{sec:lcs} present the details of BMC. Section~\ref{sec:exp} gives some basic numerical examples.
Finally, Section~\ref{sec:conclusions} draws our conclusions.

\section{Related Work}
\label{sec:related}

\Paragraph{Additive channels.}
In {\em additive channels}, a collision of $k$ packets is viewed by the receiver as a linear combination of these $k$ packets. Here the linear combination is usually defined over the individual symbols in the original packets, with vector arithmetic operations. In such a context, researchers have developed various interesting designs that can recover the $k$ original packets from $k$ collisions.
For example, Collision-Resistant Multiple Access (CRMA)~\cite{li2011crma} uses {\em network coding} in additive channels. In CRMA, the receiver obtains $k$ collisions, and solves the $k$ corresponding linear combinations to recover the $k$ original packets. CRMA uses random coefficients to make the $k$ collisions linearly independent. As another example, ZigZag decoding~\cite{gollakota2008zigzag} also obtains $k$ collisions and then solves for the $k$ original packets, while (conceptually) using random initial delays at each sender to make the $k$ collisions linearly independent.


Theoretically, those schemes~\cite{gollakota2008zigzag,li2011crma} could potentially also help to  achieve $\mathbf{R} = \Theta(1)$ in our context. However, BMC and those schemes~\cite{gollakota2008zigzag,li2011crma} target different kinds of wireless networks. First, BMC works for low-complexity physical layer implementations --- for example, even for the bare-bone OOK physical layer in Zippy~\cite{sutton15}. CRMA and ZigZag decoding instead require a receiver's radio hardware to at least be able to
estimate the coefficients (expressed as complex numbers) of the channels with different senders, and also to process a stream of complex symbols (measured every sampling interval) based on the channel coefficients.
Second, CRMA and ZigZag decoding fundamentally rely on the receiver obtaining accurate channel estimation for all the senders, so that the receiver can determine the coefficients in the linear combinations. As a result, they usually consider a rather small number of simultaneous senders. For example, Zigzag decoding~\cite{gollakota2008zigzag} focuses on 2 or 3 (rather than hundreds of) simultaneous senders. In comparison, BMC does not require such channel estimation at all. Hence BMC can work in rather dense wireless networks with many concurrent senders.
BMC can also work in networks where accurate channel estimation is simply infeasible due to fast-changing channel conditions.



\Paragraph{XOR channels.}
In {\em XOR channels}, colliding packets are XOR-ed together at the bit-level. For XOR channels, researchers have designed various codes~\cite{censor-hillel12,censor-hillel15} to enable the receiver to
recover the $k$ original packets from $k$ collisions. The schemes~\cite{censor-hillel12,censor-hillel15} for XOR channels need the receiver to be able to tell whether the number of senders sending the ``1'' bit is even or odd, which can be rather difficult to implement. In comparison, BMC can work under the OR channel, where colliding packets are OR-ed together at the bit-level. An OR channel only needs the receiver to tell whether there is at least one sender sending the ``1'' bit.

\Paragraph{All-to-all and one-to-all communication.}
BMC targets \linebreak {\tt all-to-neighbors} communication in multi-hop wireless networks, where every node wants to send a (small) data item to all its neighbors. Related to this, there have been interesting works targeting {\tt all-to-all} and {\tt one-to-all} communication in multi-hop wireless networks.

In {\tt all-to-all} communication, every node has some data item to be disseminated to all nodes in the network.
Works on {\tt all-to-all} communication (e.g., Chaos~\cite{landsiedel13}, Mixer~\cite{herrmann18}, Codecast~\cite{mohammad18}) usually exploit i) {\em network coding} for increasing packet diversity, and ii) {\em capture effect} for alleviating the collision problem. Here network coding is done on individual nodes (potentially in software), and is fundamentally different from network coding over additive channels as discussed earlier. Such network coding does not apply to {\tt all-to-neighbors} communication, where different nodes need to receive different sets of data items.
The capture effect only works when the number of concurrent senders is small~\cite{landsiedel13,herrmann18}, and
only enables the packet from the sender with the strongest signal to be decoded. As a result, scheduling is still needed for ensuring a small number of concurrent senders, and for ensuring ``stronger'' senders properly giving opportunities to ``weaker'' senders.
In comparison, BMC avoids the need of scheduling, and enables the decoding of packets from all concurrent senders. In this sense, incorporating BMC into those schemes for {\tt all-to-all} communication
could potentially further improve those schemes --- confirming this will be  part of our future work.

In {\tt one-to-all} communication, a single node wants to disseminate some data to all nodes. Works on {\tt one-to-all} communication (e.g., Glossy~\cite{ferrari11}, Splash~\cite{doddavenkatappa13}, and Pando~\cite{du15}) typically leverage i) {\em constructive interference} where multiple packets with the same content interfere constructively, ii) {\em tree pipelining} where nodes on different levels use different channels, and iii) applying {\em fountain codes} on each node. These techniques do not  apply to {\tt all-to-neighbors} communication. In particular, fountain code does not help in {\tt all-to-neighbors} communication, where different nodes need to receive different sets of data items.

\Paragraph{Capacity of wireless networks.}
As mentioned in Section~\ref{sec:intro}, Ghaffari et 
al.~\cite{ghaffari12} have proved that even with optimal scheduling, $\mathbf{R}$ will still approach zero as the network size increases.
Their result is purely due to the possibility of collision, and is fundamentally different from the well-known result on the capacity of wireless networks~\cite{gupta00}. The result from \cite{gupta00} is for a more complex setting, and stems not only from the possibility of collision, but also from the need to do multi-hop routing. Nevertheless, BMC might potentially also help to overcome the bounds in \cite{gupta00} --- confirming this is beyond the scope of this work.

\Paragraph{Compressive sensing, superimposed code, and group testing.}
Section~\ref{sec:intro} mentioned that BMC can be viewed as a kind of superimposed code~\cite{kautz64} and non-adaptive group testing (NAGT)~\cite{du99book}. Superimposed code and NAGT, in turn, are related to compressive sensing~\cite{foucart13}. However, there is a fundamental difference~\cite{gilbert12} between compressive sensing and superimposed code/NAGT: In compressive sensing, the superimposition is typically done with vector arithmetic operations. While in superimposed code and NAGT, the superimposition is done using the boolean OR operator.
The following will provide a thorough discussion on  existing works on superimposed code and NAGT. For space constraints, we do not further elaborate on works on compressive sensing, which are less relevant to BMC.

If one were to apply the existing superimposed code and NAGT designs~\cite{\allref} to our context, achieving $\mathbf{R}=\Theta(1)$ would require {\em exponential} decoding computational complexity with respect to $d$.
Specifically, a number of superimposed code and NAGT designs~\cite{aldridge14,atia12,chan14,cheraghchi11,sebo85,zhigljavsky03}, if applied to our context, could achieve $\mathbf{R}=\Theta(1)$. However, none of these schemes provides polynomial-time decoding algorithms. Some of these works (e.g., \cite{aldridge14,chan14,cheraghchi11}) do mention ``efficient'' decoding. But their notion of ``efficient'' means being polynomial with respective to $D$, where $D$ corresponds to the total number of possible data items in our context (i.e., $2^{8d}$).

More recently, researchers have developed a range of interesting designs~\cite{bui17,cai17,cheraghchi13,inan17,indyk10,lee16,ngo11,vem17} for
superimposed code and NAGT, with polynomial decoding complexity (with respect to $k$ and $d$). But none of these can achieve $\mathbf{R}=\Theta(1)$. Specifically, the designs in Indyk et al.~\cite{indyk10} and Ngo et al.~\cite{ngo11} can achieve $\mathbf{R}=\Theta(\frac{1}{k})$, while guaranteeing zero error in decoding. Inan et al.~\cite{inan17}'s designs focus on limiting the column and row weight in the testing matrix, while achieving $\mathbf{R}=O(\frac{1}{k})$ with zero error.
The remaining designs~\cite{bui17,cai17,cheraghchi13,lee16,vem17} allow some positive {\em error probability} $\delta$ in decoding.\footnote{Different works sometimes define $\delta$ in different ways, but in all cases, as $\delta$ decreases, the decoding results get closer to the entirely correct results.} Among these, GROTESQUE~\cite{cai17}
can achieve $\mathbf{R}=\Theta(\frac{1}{\ln k})$, while SAFFRON~\cite{lee16} and Bui et al.~\cite{bui17} can both achieve $\mathbf{R}=\Theta(\frac{1}{f(\delta)})$, with $f(\delta)$ being a function of $\delta$ as defined by some optimization problem. While $f(\delta)$ has no closed-form, it can be verified from the optimization problem in~\cite{lee16} that $f(\delta)\rightarrow \infty$ as $\delta\rightarrow 0$.
The design in Vem et al.~\cite{vem17} can achieve $\mathbf{R}=
d/(f(\delta)\ln \frac{f'(\delta) 2^d}{f(\delta)})$, where $f'(\delta)$ is also a  function of $\delta$. They did not obtain asymptotic bounds for $f(\delta)$ and $f'(\delta)$ when $\delta\rightarrow 0$. Finally, Cheraghchi~\cite{cheraghchi13} proposes a number of schemes while focusing on dealing with noise. None of the schemes from \cite{cheraghchi13} with polynomial decoding complexity can achieve $\mathbf{R}=\Theta(1)$.

Some superimposed code and NAGT designs~\cite{barg17,bui13,cheraghchi09,gilbert12,gilbert08,mazumdar12,mazumdar14,mazumdar16,porat11} are not explicitly concerned with computational overhead, but nevertheless may allow polynomial-time decoding. But none of these schemes can achieve $\mathbf{R}=\Theta(1)$.

Compared to all the above designs, BMC achieves $\mathbf{R}=\Theta(1)$ while needing only polynomial encoding/decoding complexity with respect to $k$, $d$, and $\delta$.
To achieve this, our design of BMC is different from the mainstream approaches for superimposed code and NAGT. For example, many existing designs either use a random testing matrix (e.g., \cite{aldridge14,atia12,chan14,cheraghchi11}) or rely on code concatenation (e.g., \cite{bui17,indyk10,lee16,ngo11,vem17}).
BMC uses neither approach. Instead, BMC first encodes the data item into a codeword using some erasure code, and then uses a low collision set (LCS) to schedule the transmission time of each symbol in this codeword. While BMC also uses Reed-Solomon (RS) code~\cite{lin04ecc} as some code-concatenation-based designs~\cite{bui17,indyk10,lee16,ngo11,vem17}, BMC leverages RS code's ability to tolerate erasures, rather than RS code's minimum distance.

\Paragraph{LCS as a stand-alone design.}
Our LCS itself can also be viewed as a {\em stand-alone} design for superimposed codes and NAGT.
But since we use LCS to schedule the transmission time of RS symbols, different elements in our LCS need to have sufficient non-overlap. Hence LCS is more related to the notion of error-correcting NAGT~\cite{macula97},
especially to its various relaxed versions~\cite{atia12,chan14,cheraghchi11,cheraghchi13,ngo11,zhigljavsky03}. Some of these relaxed versions~\cite{cheraghchi13,ngo11} (also called {\em error-correcting list-disjunct matrices}) allow the decoding result to contain at most $l$ false positives (but no false negatives), deterministically. In comparison, LCS achieves no false positive/negative across all the $k$ senders, with $1-k\delta$ probability. Some other versions~\cite{atia12,chan14,cheraghchi11,zhigljavsky03} offer similar probabilistic guarantees as LCS. But different from the designs in \cite{atia12,chan14,cheraghchi11,zhigljavsky03}, in LCS each element has a fixed weight, so that it can be used to schedule the transmission of RS symbols. Furthermore, \cite{atia12,chan14,cheraghchi11,zhigljavsky03} only shows that a randomly constructed matrix provides the desirable property with high probability (or on expectation). In comparison, we derive some sufficient condition for a randomly constructed set to be an LCS, and such sufficient condition can be verified in polynomial time.

\ifthenelse{\boolean{shortversion}}{
\Paragraph{Scheduling via superimposed codes.}
There is also a large body of works on packet scheduling~\cite{berger84,bonis17,komlos85,wolf85} and channel assignment~\cite{xing07} via superimposed codes.
Since this approach~\cite{berger84,bonis17,komlos85,wolf85} is still doing packet-level scheduling, it cannot overcome the fundamental barrier of $\mathbf{R} = O(\frac{1}{\ln N})$ from \cite{ghaffari12}.
}{
\Paragraph{Scheduling via superimposed codes.}
There is also a large body of works on packet scheduling~\cite{berger84,bonis17,komlos85,wolf85} and channel assignment~\cite{xing07} via superimposed codes. Here each sender, conceptually, has a packet to send.
The idea is that each sender will correspond to a (distinct) codeword.
The positions of the ``1'' bits in the codeword determine in which intervals (or on which channels) the corresponding sender should send its packet. Note that the codeword here has nothing to do with the content of the sender's packet.
A sender may send the {\em same} packet multiple times if there are multiple ``1'' bits in its codeword.
The hope is that for every sender, there exists at least one interval (or one channel) during which the sender is the only one sending.
Since this approach~\cite{berger84,bonis17,komlos85,wolf85} is still doing packet-level scheduling, it cannot overcome the fundamental barrier of $\mathbf{R} = O(\frac{1}{\ln N})$ from \cite{ghaffari12}.
}

\Paragraph{Public/private coins.}
In designing BMC, one of the ideas we use is to reduce the number of different ways to select the transmission slots. This follows the spirit of using private coins to simulate public coins~\cite{newman91}. However, the original mechanism from \cite{newman91} is only an existential proof, while BMC obviously needs to construct an explicit protocol.

\section{Overview of BMC}
\label{sec:overview}

Recall the example in Figure~\ref{fig:app} where each receiver/rescuer needs to receive a $d$-byte data item from each of its neighboring \linebreak senders/sensors (i.e., senders within the receiver's communication range).
To facilitate understanding, assume for now that each receiver has the same number $k$ of neighboring senders (note that this assumption is not actually needed for our BMC protocol or its formal guarantees). Our goal is\mbox{} to enable every receiver to receive its respective $k$ data items in $\Theta(kd)$ time, regardless of the network topology (i.e., which senders are neighbors of which receivers).

Let $N$ be the total number of nodes in the network, including both senders and receivers.
We will assume $k =  \omega(\ln N)$, since scheduling tends to be harder as $k$ increases, and we want BMC to address the harder cases. (For readers unfamiliar with the $\omega()$ notation: If $k =  \omega(\ln N)$, then $k = \Omega(\ln N)$ and $k\ne \Theta(\ln N)$.)
%
We also need $d$ not to be too small so that our formal analysis later can approximate the tails of the various distributions --- specifically, we assume $d = \omega(\ln^2 N\times \ln\ln N)$.
We expect such a condition to be relatively easy to satisfy, since the terms on the right-hand side are logarithmic, and since one could concatenate multiple data items together to increase $d$, if needed.

\subsection{Sharing the Damage of Collision}

A careful look at the $\Omega(kd\ln N)$ barrier~\cite{ghaffari12} leads to a basic observation: {\em Implicitly, scheduling in wireless networks is always done at the packet-level, with packets being the units for scheduling}. This means that in our scenario, if a sender's packet (containing its $d$-byte data item) does not collide with other packets, then most of the bits in the packet will be intact. But if it does, then many of its bits will be affected. Assume as an example that $20$\% of the packets experience collision, where the receivers of these packets are not able to decode them correctly.

Now instead of having $20$\% such unlucky packets, what if all the packets share the ``damage of collision''? This means that about $20$\% of the bits in {\em every} packet will be corrupted. Quite interestingly, doing so shifts the paradigm: We can easily use proper coding to tolerate those $20$\% errors in each packet, and successfully recover {\em all} packets correctly.

\Paragraph{Cross-sender bit-mixing.}
The above forms the starting point of our BMC design. Conceptually, BMC partitions the $\Theta(kd)$ available time into $\Theta(kd)$ {\em slots} where each slot is the airtime of, for example, a single bit (Figure~\ref{fig:bmc1}).
Among these slots, each sender chooses $\Theta(d)$ slots in a certain randomized fashion, {\em without coordinating with other senders}. Next,
each sender embeds $\Theta(d)$ bits of information for its own data item into those chosen slots. The remaining slots are left ``blank''. This then becomes a BMC codeword for that sender.
With slight abuse of notation, now a bit in a BMC codeword may take one of the following three values: ``0'', ``1'', or ``blank''. We will later explain how to do modulation/demodulation for blank bits.

Our BMC design includes a method for choosing the slots, so that with good probability and without needing any coordination, among a sender's $\Theta(d)$ chosen slots, a majority of them do not collide with other senders' choices. BMC can infer which slots suffer from collisions, and will then treat those slots simply as erasures. In some sense, one could view the selection of the $\Theta(d)$ slots as a form of ``bit-level scheduling'', as compared to standard packet-level scheduling. Such ultra-fine-grained ``bit-level scheduling'' results in bits from different senders being mixed together in BMC.

\begin{figure}
  \centering
  \includegraphics[width=\columnwidth]{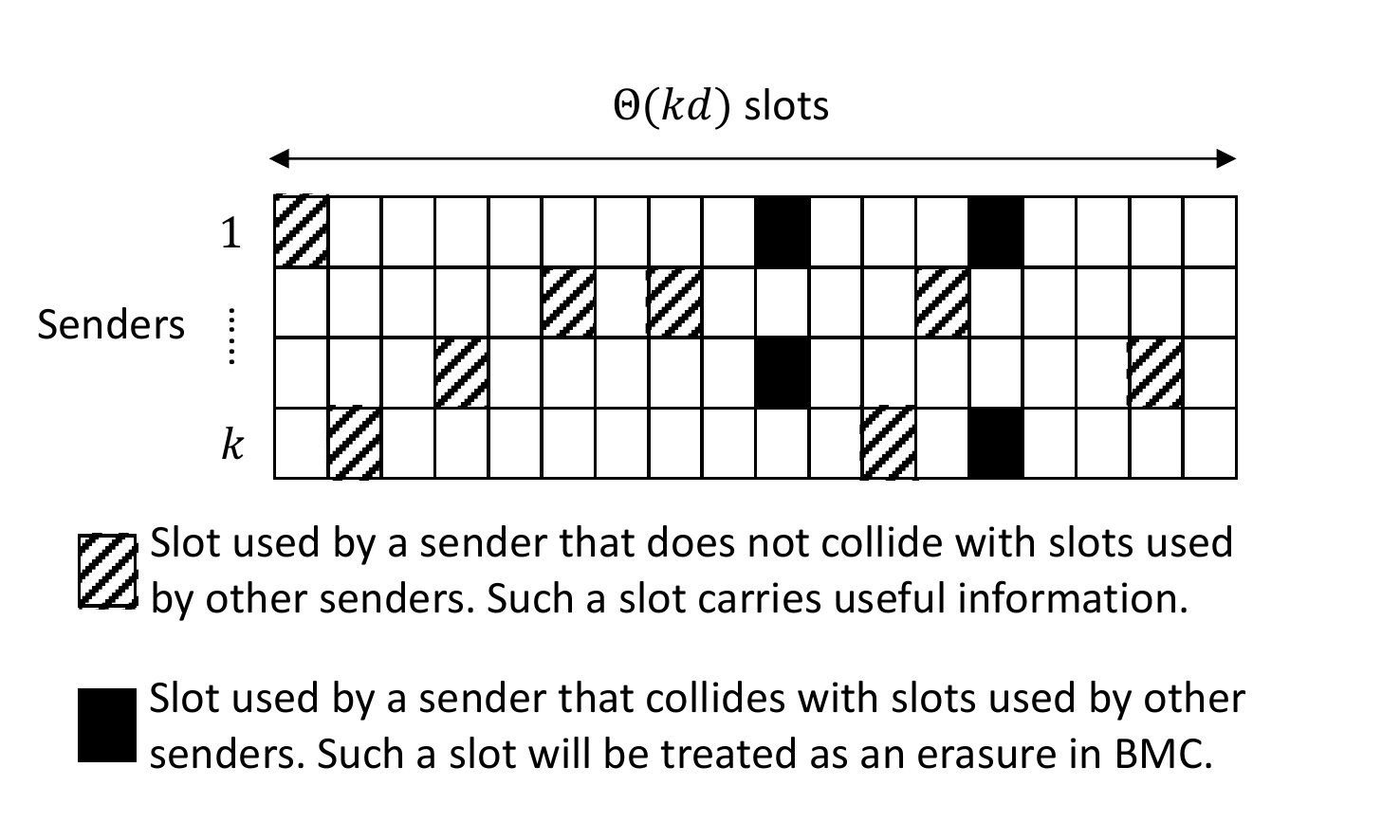}
  \vspace{-8mm}
  \caption{Cross-sender bit-mixing.}
  \vspace{-2mm}
  \label{fig:bmc1}
\end{figure}


\subsection{Central Challenge in Bit-Mixing}

While the idea of ``bit-level scheduling'' is conceptually simple, it introduces a new and unique challenge that was not present in standard (packet-level) scheduling: Unlike a packet, a bit does not have (or cannot afford to have) a ``header''. Hence a receiver cannot easily tell which bits are from which senders.
This constitutes the central challenge in BMC: To decode, a receiver needs to know which $\Theta(d)$ slots are chosen by each sender. Since each sender chooses $\Theta(d)$ slots out of total $\Theta(kd)$ slots, it may take up to $\log_2 {\Theta(kd) \choose \Theta(d)} = \Omega(d\log k)$ bits to describe those slots. This is even larger than the $d$-byte data item itself.

\Paragraph{Constraining choices.}
Our first step in overcoming this challenge is to substantially reduce the number of possible ways to do such selections. Formally, we use a {\em masking string} to specify which $\Theta(d)$ slots, out of the $\Theta(kd)$ slots, are selected. The masking string has a length of $\Theta(kd)$ and contains only ``1'' bits and ``blank'' bits, where a ``1'' bit means that the corresponding slot is selected.
We will construct a set $S$ with only $\Theta(\frac{k}{\delta})$ masking strings, with certain properties (Figure~\ref{fig:bmc2}). Here $\delta$ is a tunable parameter in BMC. It corresponds to the probability of delivery failure of a data item, and is usually a small value such as $o(\frac{1}{k})$ or $o(\frac{1}{kN})$. Hence the set $S$ will usually have $\omega(k^2)$ or $\omega(k^2 N)$ masking strings. Roughly speaking, setting $\delta = o(\frac{1}{k})$ will ensure that for any {\em given} receiver, the probability of it successfully decoding all its $k$ data items is close to $1$. Since there can be up to $N$ receivers in the network, having $\delta = o(\frac{1}{k})$ would mean that while most receivers can decode successfully, there may still be a vanishingly small fraction of receivers that cannot. In comparison, setting $\delta = o(\frac{1}{kN})$ will provide an even stronger guarantee: With probability close to $1$, {\em all} receivers in the networks will successfully decode {\em all} their respective $k$ data items.

A sender will choose a uniformly random string from $S$, and then use those slots as specified by the string. Doing so decreases the number of ways to select the slots from ${\Theta(kd) \choose \Theta(d)}$ to $\Theta(\frac{k}{\delta})$.\footnote{Note that even though $\delta$ is a small value such as $o(\frac{1}{k})$ or $o(\frac{1}{kN})$, the quantity of $\frac{k}{\delta}$ will still be much smaller than ${kd \choose d}$, since ${kd \choose d} > k^d$ where $d$ is on the exponent.}
As a critical step, we will prove that constraining ourselves to the masking strings in $S$ will not disrupt the properties that we need: Namely, with good probability, among a sender's $\Theta(d)$ chosen slots, a majority of them will not collide with other senders' choices.

\begin{figure}
  \centering
  \hspace*{-9mm}
  \includegraphics[scale=.4]{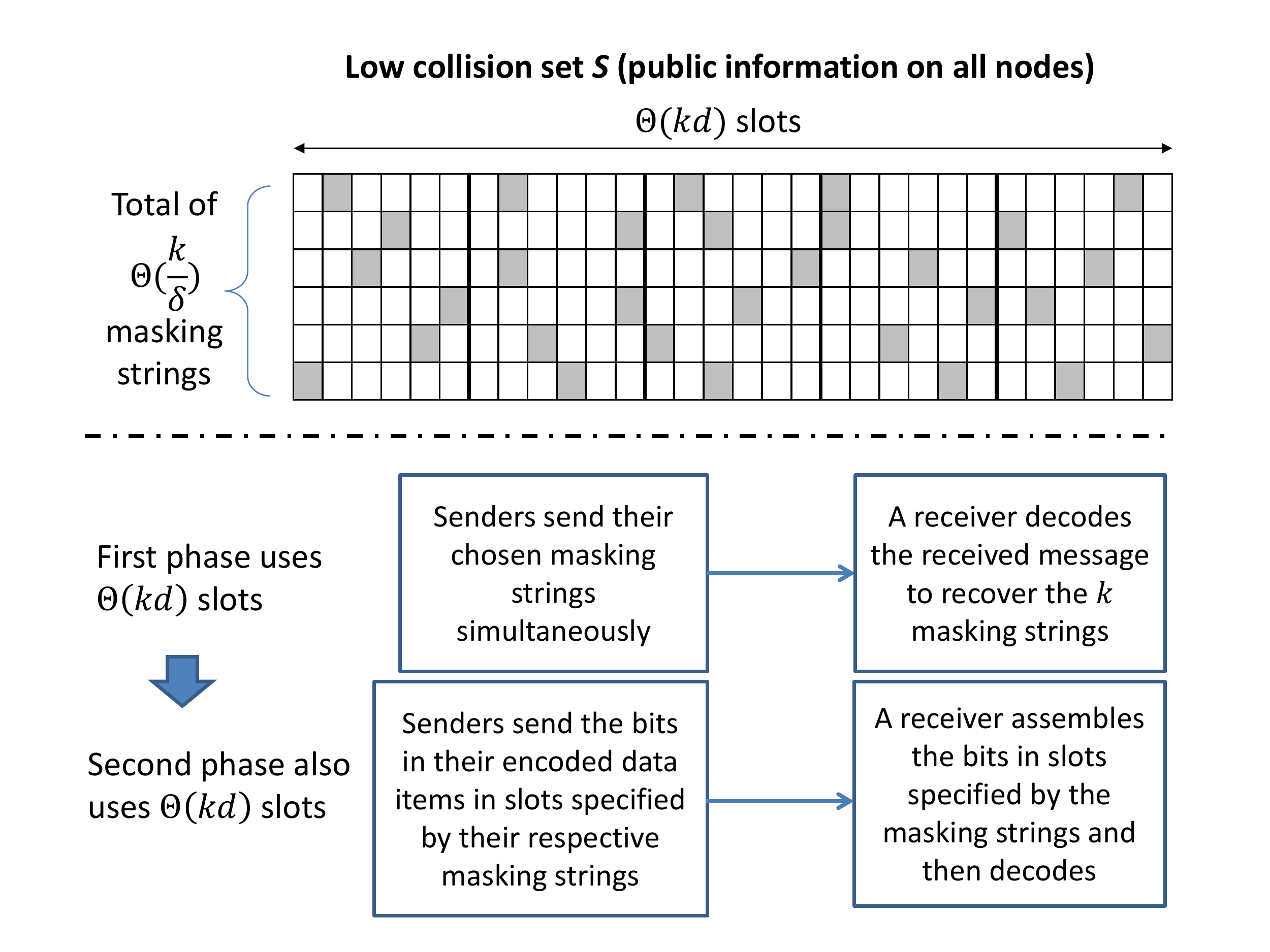}
  \vspace{-8mm}
  \caption{The two phases in BMC, with $k=  \omega(\ln N)$ and $d = \omega(\ln^2 N\times \ln\ln N)$.
  }
  \vspace{-5mm}
  \label{fig:bmc2}
\end{figure}

\Paragraph{Determining which masking strings are chosen.}
We will have all parties keep a copy of $S$ --- namely, $S$ will be hardcoded. But the senders still needs to communicate to a receiver which masking strings are being used.
Doing so naively would bring us back to the problem of packet collision and scheduling.

In BMC, to inform a receiver which masking strings are being used,
prior to sending the data items,
all the senders will first send their respective masking strings. This results in two conceptual phases --- see Figure~\ref{fig:bmc2}. Note that these two phases are only conceptual, and do not involve two interactive rounds. In both phases, the senders will  send {\em simultaneously} and in a bit-aligned fashion.

We will prove that our set $S$ has an additional property: With probability $1-\delta$, for all $\lambda\in S$ where $\lambda$ is not used by any of the $k$ senders, at most half of the $\Theta(d)$ slots chosen by $\lambda$ may collide with the masking strings used by the $k$ senders. This property enables a receiver to decode the masking strings in the following way.\footnote{This is only for decoding the masking strings --- the data items will be decoded separately in the second phase.} The receiver will compare the message $z$ it receives, with {\em every} masking string $\lambda\in S$:
\begin{itemize}
\item If $\lambda$ was used/sent by some sender, then in {\em all slots} where $\lambda$ is ``1'', the message $z$ should have a ``1'' bit or a collision of multiple ``1'' bits. (Remember that no sender sends ``0'' bits in the first phase.)
\item If $\lambda$ was not used/sent by any sender, then in {\em at most half of the slots} where $\lambda$ is ``1'', the message $z$ will have a ``1'' bit or a collision of multiple ``1'' bits. (This is due to the additional property of $S$ described above.)
\end{itemize}
Such a separation enables a receiver to tell whether $\lambda$ was used/sent by one of the $k$ senders.\footnote{Section~\ref{sec:practical} will further discuss how to deal with errors in transmission.}

\Paragraph{Achieving $\mathbf{R}=\Theta(1)$.}
Let us quickly summarize how the above design achieves $\mathbf{R}=\Theta(1)$. Each receiver needs to receive $k$ data items from its respective $k$ senders, where $k =  \omega(\ln N)$. Each data item has $d$ bytes, where $d = \omega(\ln^2 N\times \ln\ln N)$. The length of each masking string is $\Theta(kd)= \omega(\ln^3 N\times \ln\ln N)$. 

Each sender does the following, without worrying about how many receivers it corresponds to: Each sender first sends its chosen masking string, taking $\Theta(kd)$ slots, or equivalently, $\Theta(kd)$ bits of airtime. Next each sender sends its encoded data item, again taking $\Theta(kd)$ bits of airtime. Note that all senders simultaneously go through these two phases synchronously, hence the total airtime needed is just $\Theta(kd)$ bits. Each receiver, upon successful decoding, obtains $k$ data items, each with $d$ bytes. Hence we have the medium utilization rate $\mathbf{R}=\frac{k\cdot 8d} {\Theta(kd)} = \Theta(1)$.

\section{Physical Layer Issues}
\label{sec:assumption}

\subsection{BMC's Assumptions}

BMC needs a few assumptions on the physical layer.

\Paragraph{Synchronization.}
BMC assumes that in the first phase, all the senders send their masking strings synchronously, so that the packets are superimposed in a bit-aligned fashion. Similarly in the second phase, all the senders send their encoded data items in a synchronized and bit-aligned fashion.

\Paragraph{Modulation/demodulation.}
A bit in a BMC codeword or masking string may be ``0'', ``1'', or ``blank''. For modulation, when a sender sends a ``blank'' bit, BMC assumes that the sender does not emit radio signals.\footnote{BMC allows the modulation of the ``0'' bit to be the same as that of the ``blank'' bit. BMC never needs to differentiate a ``0'' bit from a ``blank'' bit in demodulation.}

In the BMC protocol, in any given slot, there are $k$ senders each sending a bit, with each bit being ``0'', ``1'', or ``blank''. But not all combinations of ``0'', ``1'', and ``blank'' bits are possible in a given slot. For example, in any slot in the first phase of the protocol, either all $k$ senders send ``blank'' bits, or less than $k$ of them send ``blank'' bits and all the remaining ones send ``1'' bits. Hence the receiver has the {\em prior knowledge} that it must be one of these two cases.
Our following assumptions need to hold only when such prior knowledge is available to the receiver. Specifically, we assume that in {\em any given slot}:
\begin{enumerate}
\item Given the prior knowledge that exactly one sender sends a non-``blank'' bit and all other senders send ``blank'' bits, the receiver can tell whether the non-``blank'' bit is ``0'' or ``1''. (This property is needed for the second phase of BMC.)
\item Given the prior knowledge that either i) all $k$ senders send ``blank'' bits or ii) less than $k$ of them send ``blank'' bits and all the remaining ones send ``1'' bits, the receiver can distinguish these two cases. Without loss of generality, we denote the demodulation result in these two cases as ``0'' and ``1'', respectively. (This property is needed for the first phase of BMC.)
\end{enumerate}
Note that in BMC, the receiver will always have the respective prior knowledge (directly from the BMC protocol) whenever it needs to satisfy the above assumptions. BMC does not need any other assumptions on (de)modulation. For example, BMC is not concerned with the demodulation of the collision of multiple ``0'' bits, or the collision of ``0'' bits and ``1'' bits, since such collisions can only occur in those slots not used by BMC decoding.

\subsection{Using BMC with Some Example Physical-layer Implementations}

The following discusses how BMC can be potentially used with some example physical layer implementations.


\Paragraph{Using BMC with Zippy's physical layer.}
Zippy~\cite{sutton15} is a recent design for on-demand flooding in multi-hop wireless networks. Its physical layer uses OOK modulation to simplify the transceiver circuitry and to achieve superior power-efficiency. BMC can be used over Zippy's physical layer, without any changes needed to the physical layer.

Specifically, with Zippy's OOK modulation, ``blank'' bits in BMC should be directly treated as ``0'' bits in modulation. Hence a sender will not emit radio signals for any ``blank'' bit or ``0'' bit. Now if in a slot exactly one sender sends an information bit (i.e., a ``0'' or ``1'' bit) and all other senders send ``blank'' bits, obviously a receiver in Zippy can tell whether the information bit is ``0'' or ``1''. Next, if in a slot one or more senders send ``1'' bits while the remaining senders send ``blank'' bits, Zippy~\cite{sutton15} has shown that with {\em carrier frequency randomization}, the receiver can effectively demodulate the received signal to ``1''. Thus the receiver can properly differentiate the case where all $k$ senders send ``blank'' bits from the case where less than $k$ of them send ``blank'' bits and all the remaining ones send ``1'' bits. Hence BMC's two assumptions on demodulation are both satisfied.

For synchronization among senders, BMC could directly use the existing distributed synchronization mechanism in Zippy~\cite{sutton15}. Zippy~\cite{sutton15} has shown that it can achieve a synchronization error of tens of microseconds between all pairs of  neighboring nodes, throughout the network. Since Zippy operates on a slow data rate of $1.36$ kbps with each bit taking about $700$ microseconds,
and since a receiver takes multiple samples for each bit, such an error should already enable good bit-level alignment.
Finally, due to clock drift on each node, re-synchronization will be needed periodically. Since a typical crystal oscillator can drift 20 parts per million (PPM), re-synchronization can be done once every few seconds.
Since Zippy's network-wide synchronization takes only tens of milliseconds~\cite{sutton15}, the fraction of the airtime wasted by such periodic re-synchronization will just be a few percent.

\Paragraph{Using BMC in RFID systems.}
In RFID systems, an {\em interrogator} transmits a radio wave to {\em tags}, and each tag either reflects the radio wave back (which corresponds to sending back a ``1'' bit) or keeps silent (which corresponds to sending back a ``0'' bit). The modulated backscattered wave can then be demodulated by one or multiple receivers.

BMC can be used in single-interrogator RFID systems without needing any changes to the physical layer. Specifically, with backscatter communication in RFID systems, synchronization is already achieved, and the bits sent back from the tags will already be properly aligned. The tags (i.e., senders) will treat ``blank'' bits the same as ``0'' bits. The two assumptions needed by BMC on demodulation will then be directly satisfied~\cite{zheng14}. When there are multiple interrogators, to use BMC, the interrogators need to first properly synchronize among themselves (e.g., via a backhaul network) so they transmit the same radio wave synchronously.

\Paragraph{Using BMC in ZigBee systems.}
BMC might find its applicability in more complex wireless systems as well, after appropriate changes to the physical layer. Let us take ZigBee (IEEE 802.15.4)  as an example. ZigBee may use DSSS/O-QPSK modulation in the 2.4GHz band to transmit {\em ZigBee symbols}. Each ZigBee symbol contains 32 {\em chips}, which map to 4 bits.

First, to achieve the synchronization needed by BMC in ZigBee, one could use the distributed synchronization mechanism in Glossy~\cite{ferrari11}. Under ZigBee, Glossy achieves a synchronization error of less than 0.5 microsecond among neighbors~\cite{ferrari11}.
At 250kbps data rate, each ZigBee symbol takes 16 microseconds. With some extra inter-symbol guard time, we expect such synchronization error to be small enough to achieve good symbol alignment across the senders.
As before, periodic re-synchronization may be needed due to clock drifts.
For example, with 2.5-microsecond inter-symbol guard time and using oscillators with maximum 20 PPM clock drift, it suffices to re-synchronize every 0.1 second.
Glossy achieves distributed synchronization by flooding a packet. If each hop in the flooding takes 0.5 millisecond, we estimate that Glossy's flooding will likely finish within 3 millisecond in a 5-hop network. Hence we estimate the fraction of the airtime wasted by re-synchronization to be roughly $(\mbox{$3$ milliseconds})/(\mbox{$0.1$ second}) = 3\%$ in such a case.


Second, to satisfy the assumptions needed by BMC on demodulation, one could
introduce ``blank'' ZigBee symbols (instead of ``blank'' bits). When sending a ``blank'' ZigBee symbol, the sender just keeps silent.

In the first phase of BMC, when the protocol needs to send a ``blank'' bit (or ``1'' bit), we will actually let the sender send a ``blank'' ZigBee symbol (or a ZigBee symbol corresponding to  ``1111'').\footnote{This will make the first phase less efficient, but note that BMC's airtime will likely be dominated by the second phase anyway.}
Recall that the synchronization error between different senders is supposed to be well below the duration of a ZigBee symbol. We hence expect that the receiver can differentiate a
``blank'' ZigBee symbol from the superimposition of one or more ZigBee symbols that all correspond to  ``1111'', by examining the energy level of the received signal.
This then satisfies the demodulation assumption needed for the first phase of BMC.

In the second phase of BMC, in our full design (see Section~\ref{sec:codingdata}), a sender actually sends a Reed-Solomon code symbol (RS-code symbol) in each slot. An 8-bit RS-code symbol then corresponds to
two 4-bit ZigBee symbols, and a ``blank'' RS-code symbol conveniently translates to two ``blank'' ZigBee symbols. During the second phase, BMC will require the receiver to demodulate ZigBee symbols from {\em different} senders.
To enable such demodulation in DSSS/O-QPSK, we may need to add a few reference chips before every two ZigBee symbols (i.e., every RS-code symbol) to recalibrate the demodulation baseline for the sender of the next two ZigBee symbols. (One could further optimize by sending such reference chips only when needed, instead of for every two ZigBee symbols.)

\section{BMC Encoding and Decoding}
\label{sec:details}

This section will elaborate BMC encoding/decoding algorithm. Our BMC algorithm critically relies on the existence of a {\em low collision set} (or {\em LCS}). The existence of LCS, as well as the possibility of finding one, will be formally proved in Section~\ref{sec:lcs}.

Table 1 summarizes our key notations.
BMC assumes that the maximum degree $k$ of a receiver in the wireless network is known. In practice, it suffices to provide BMC with some upper bound $k'$ for $k$. The only consequence is that the resulting $\mathbf{R}$ will be reduced by a factor of $\frac{k'}{k}$.
BMC also assumes that the network size $N$ is known. Again, in practice, it suffices to provide BMC with some upper bound $N'$ for $N$. The only consequence is that the time complexity and space complexity of BMC may increase by a factor of $\frac{N'}{N}$.


\begin{table}
\centering
\caption{Key notations.}
\vspace*{-4mm}
\hspace*{-2mm}
\begin{tabular}{|c|l|}
\hline
$N$ & total number of nodes in the network \\
\hline
$k$ & maximum degree of a receiver in the network\\
& (e.g., 100 in practice, and $\omega(\ln N)$ asymptotically) \\
\hline
$d$ & size (in bytes) of the data item (with CRC) on each sender\\
&  (e.g., 100 in practice, and $\omega(\ln^2 N\times \ln\ln N)$ asymptotically) \\
\hline
$w$ & weight of masking string \\
&(e.g., $w=2d$ for $1$-byte RS symbols) \\
\hline
$\delta$ & tunable parameter in BMC (e.g., $o(\frac{1}{k})$ or $o(\frac{1}{kN})$), \\
& corresponding to delivery failure probability  \\
\hline
\end{tabular}
\vspace*{-2mm}
\end{table}

\subsection{Low Collision Set}
\label{sec:lcsdef}

We first define {\em masking strings}. In the previous section, we explained that a masking string only contains ``1'' bits and ``blank'' bits. For ease of discussion, from this point on, we will use ``0'' bits to represent ``blank'' bits in the masking strings.
\begin{definition}
A binary string is a {\em $(k,w)$ masking string} if it is the concatenation of $w$ (potentially different) binary substrings of length $4k$, with each substring having a Hamming-weight of $1$.
\end{definition}

Obviously, a $(k,w)$ masking string has a length of $4kw$ and a Hamming-weight of $w$. For two equal-length binary strings $\lambda$ and $\eta$, recall that their {\em inner product} (denoted as $\lambda\cdot \eta$) is defined as $\lambda\cdot \eta = \sum_i (\lambda[i]\times \eta[i])$. (Throughout this paper, we use $a[i]$ to denote the $i$-th element of a binary string $a$.) The following defines {\em compatibility} between a masking string $\lambda$ and a multi-set $T$ of masking strings. Intuitively, if they are compatible, then the total number of collisions between $\lambda$ and $T$ is limited:
\begin{definition}
Consider any $(k,w)$ masking string $\lambda$ and any multi-set $T=\{t_1, t_2, \ldots, t_m\}$ of $(k,w)$ masking strings. We say that $T$ is {\em compatible} with $\lambda$ if and only if\, $\sum_{i=1}^{m} (\lambda \cdot t_i) \le \frac{w}{2}$.
\end{definition}

We can now define an LCS:
\begin{definition}
\label{def:lcs}
A set $S$ of $(k,w)$ masking strings is a {\em $(k,w, \delta)$ low collision set} (or {\em LCS} in short) if it satisfies the following property for all given $i$ and $m$ (where $1\le i\le m \le k$): Imagine that we choose $m$ elements (denoted as $t_1$ through $t_m$) from $S$ uniformly randomly with replacement. Then with probability at least $1-\delta$:
\begin{enumerate}
\item The multi-set $T = \{t_1,\ldots, t_m\}$ is compatible with all $\lambda\in S\setminus T$, and
\item The multi-set $T_{\bar{i}} = \{t_1, \ldots, t_{i-1}, t_{i+1}, \ldots, t_m\}$ is compatible with $t_i$.
\end{enumerate}
\end{definition}

\subsection{Encoding/Decoding Masking Strings}
\label{sec:codingmasking}

BMC has separate encoding/decoding algorithms for masking strings (Algorithm~\ref{alg:masking}) and for data items (Algorithm~\ref{alg:data}).
The senders and the receivers will first invoke Algorithm~\ref{alg:masking} and then invoke Algorithm~\ref{alg:data}.

\begin{algorithm}
\caption{\label{alg:masking}
Encoding/decoding of masking strings. $S$ is an LCS of size $\frac{2k}{\delta}$, where $\delta$ is a tunable parameter.}
\small
\begin{algorithmic}[1]
\Statex \hspace*{-6mm} {\em Encoding algorithm}
\State $\lambda\leftarrow $ A uniformly random element from the set $S$;
\State return $\lambda$ after replacing ``0'' bits in $\lambda$ with ``blank'' bits;
\Statex
\setcounter{ALG@line}{0}
\Statex \hspace*{-6mm} {\em Decoding algorithm}
({\bf input:} A received binary string $z$ of $4kw$ bits;
{\bf output:} A list of masking strings)
\State {\bf foreach} $\lambda\in S$ {\bf do}
\State \hspace*{4mm} {\bf if} $\lambda\cdot z \ge \frac{3w}{4}$ {\bf then} output $\lambda$;
\end{algorithmic}
\normalsize
\end{algorithm}

Algorithm~\ref{alg:masking} has an LCS $S$ of size $\frac{2k}{\delta}$ hardcoded into it. The algorithm has each sender select a uniformly random masking string from $S$, and then send to the receiver. The decoding part does an exhaustive enumeration of all $\lambda$ in $S$. As long as the inner product of $\lambda$ and the received string $z$ is at least $\frac{3w}{4}$, the algorithm will claim that $\lambda$ has been sent by some sender.\footnote{The algorithm does not intend to determine the id of the sender --- the id (if needed) can be included as part of the data item in Algorithm~\ref{alg:data}.} Note that in each slot in this algorithm, the receiver has the prior knowledge that either i) all senders send ``blank'' bits or ii) some of them (potentially none) send ``blank'' bits while all the remaining ones send ``1'' bits.


\subsection{Encoding/Decoding Data Items}
\label{sec:codingdata}

\begin{algorithm}
\caption{\label{alg:data}
Encoding/decoding of data items.}
\small
\begin{algorithmic}[1]
\Statex \hspace*{-6mm} {\em Encoding algorithm}
({\bf input:} A data item;
{\bf output:} A codeword)
\State encode the data item into $w$ RS symbols, with a coding rate of $\frac{1}{2}$;
\State let $x$ be the resulting RS codeword, and let $\lambda$ be the masking string returned by the encoding part in Algorithm~\ref{alg:masking};
\State $\tau\leftarrow$ empty string;
\State {\bf for $i$ from $1$ to $4kw$ do}
\State \hspace*{4mm} {\bf if} $\lambda[i] =1$ {\bf then}
remove $x$'s first RS symbol, and append it to $\tau$;
\State \hspace*{4mm} {\bf else} append a ``blank'' RS symbol to $\tau$;
\State return $\tau$;
\Statex
\setcounter{ALG@line}{0}
\Statex \hspace*{-6mm} {\em Decoding algorithm}
({\bf input:} A received string $z$ with $4kw$ RS symbols, and a list $T$ of decoded masking strings returned by Algorithm~\ref{alg:masking};
{\bf output:} A list of data items)
\State {\bf foreach} $\lambda\in T$ {\bf do}
\State \hspace*{4mm} $x \leftarrow$ empty string;
\State \hspace*{4mm} {\bf for $i$ from $1$ to $4kw$ do}
\State \hspace*{4mm}\hspace*{4mm} {\bf if} ($\lambda[i] = 1$) and
 (there exists no $\lambda'\in T$
\Statex \hspace*{7mm} $\mbox{ such that } \lambda' \ne \lambda \mbox{ and } \lambda'[i] =1$) {\bf then}
\State \hspace*{4mm}\hspace*{4mm}\hspace*{4mm}
append the $i$-th RS symbol in $z$ to $x$;
\State \hspace*{4mm} {\bf endfor}
\State \hspace*{4mm} $y\leftarrow $ RS decoding result of $x$;
\State \hspace*{4mm} {\bf if} CRC check passes on $y$ {\bf then} output $y$;
\State {\bf endfor}
\end{algorithmic}
\normalsize
\end{algorithm}

\Paragraph{Encoding.}
Algorithm~\ref{alg:data} is for encoding/decoding data items. A sender first computes a CRC on its original data item. From this point on in this paper, whenever we refer to a ``data item'', we include its CRC. The data item will then be encoded using Reed-Solomon (RS) code~\cite{lin04ecc} with a coding rate of $\frac{1}{2}$, into total $w$ RS symbols.
Here the value of $w$ and the RS symbol size $u$ (in bytes) should satisfy $2\cdot d\le w\cdot u$, so that the RS codeword is sufficiently long to accommodate the encoded $d$-byte data item. The values of $w$ and $u$ should also satisfy the inherent constraint~\cite{lin04ecc} of $w\le 2^{8u}-1$ in RS codes. For any given $d$, there are actually infinite number of $(w,u)$ pairs satisfying the above two requirements.
Among all such pairs, BMC chooses the $(w,u)$ pair with the smallest $u$ value, with tie-breaking favoring smaller $w$. This gives us a unique $(w,u)$ pair for the given $d$. One can easily verify that for $d = \omega(\ln^2 N\times \ln\ln N)$, our chosen $w$ must be $\omega(\ln^2 N)$ and chosen $u$ will be $\Theta(\ln w)$.

Recall that a sender has already chosen a masking string of $4kw$ bits in Algorithm 1. With this masking string and with the RS codeword constructed above, a sender constructs a new string $\tau$ with total $4kw$ RS symbols. There are exactly $w$ locations where $\lambda$ has the ``1'' bit. The sender embeds the $w$ RS symbols from the RS codeword into those corresponding $w$ locations of $\tau$. For each of the remaining locations, $\tau$ will have a ``blank'' RS symbol consisting of $8u$ ``blank'' bits.

\Paragraph{Decoding.}
A receiver will receive a string $z$ with total $4kw$ RS symbols. Note that the receiver already has a list $T$ of masking strings, as output by Algorithm~\ref{alg:masking}. For each $\lambda \in T$, the receiver tries to decode the corresponding data item (Line 2 to 8 in the decoding part of Algorithm~\ref{alg:data}). Given $\lambda$, the algorithm will include the $i$-th RS symbol in $z$ for the purpose of RS decoding, iff $\lambda$ is the only masking string in $T$ that has a ``1'' bit in the $i$-th location.
Note that based on $T$, the algorithm already knows at which locations $\lambda$ will ``collide'' with other masking strings. This enables the algorithm to treat the collided locations as {\em erased} RS symbols and ignore them. Doing so helps to decrease the redundancy needed in the RS code, as compared to simply treating those RS symbols as {\em erroneous}.

One can see that the above process only relies on those ``non-collision'' slots. For each such slot, the receiver has the prior knowledge that exactly one sender sends a non-``blank'' bit and all other senders send ``blank'' bits.
Finally, the CRC serves to deal with the case where Algorithm 1 returned a spurious masking string not sent by anyone. While the probability of this happening is only $\delta$, in practice, checking the CRC helps to further reduce the possibility of
Algorithm 2 returning a spurious data item due to a spurious masking string.

\subsection{Final Provable Guarantees of BMC}
\label{sec:guarantees}

\Paragraph{Achieving $\mathbf{R}=\Theta(1)$.}
Theorem~\ref{the:main} next proves that with probability at least $1-kN\delta$,
using BMC achieves $\mathbf{R}=\Theta(1)$. (Recall that $\delta$ is tunable and can be set to $o(\frac{1}{kN})$.)
This guarantee is strong in the sense that {\em all} receivers simultaneously succeed in decoding {\em all} their respective data items --- this requires our analysis to invoke a union bound across all the (up to $N$) receivers in the analysis.

The proof of the theorem actually also shows that if we consider any given receiver, then the probability of it successfully decoding {\em all} its data items using BMC will be $1-k\delta$. Hence for any given receiver, in order for this probability to approach $1$, having $\delta = o(\frac{1}{k})$ already suffices.

\begin{theorem}
\label{the:main}
Consider any wireless network with $N$ nodes, where some of the nodes are senders and the remaining ones are receivers. Each sender has a $d$-byte data item that needs to be sent to all its neighboring receivers. Let $k$ be the maximum degree (i.e. number of neighboring senders) of a receiver in the network. Let $u$ denote the Reed-Solomon symbol size (in bytes) used in Algorithm 2, and assume that $u\ge 1$. Let one {\em byte of airtime} be the time needed to transmit one byte. Then assuming the existence of a $(k,w,\delta)$ LCS of size $\frac{2k}{\delta}$ and using Algorithm 1 and 2:
\begin{enumerate}
\item Within at most $9kd$ bytes of airtime, all senders will complete their transmissions .
\item With probability at least $1-kN\delta$, all receivers in the network will output all the data items sent by their respective neighboring senders and output no other items, hence achieving $\mathbf{R}\ge \frac{1}{9}$.
\end{enumerate}
\vspace*{-3mm}
\end{theorem}
\begin{proof}
The value of $w$ in Algorithm 1 and 2 will be $w = \frac{2d}{u}$. Algorithm 1 takes $4kw = \frac{8kd}{u}$ bits (or $\frac{kd}{u}$ bytes) of airtime. Algorithm 2 sends $4kw$ RS symbols, incurring $4kw\times u = 8kd$ bytes of airtime. Hence the total airtime is $\frac{kd}{u} + 8kd = (8+\frac{1}{u})kd \le 9kd$ bytes.

We next move on to the second claim in the theorem. We will prove that for any given receiver $X$, with probability at least $1-k\delta$, it will output and only output all the data items sent by its neighboring senders. Taking a union bound across all receivers will immediately lead to the second claim in the theorem.

Let $m$ ($m\le k$) be the number of neighboring senders of $X$. For $1\le i\le m$, let $t_i$ be the maskings string chosen by neighbor $i$. For any given $i$, by the definition of LCS (Definition~\ref{def:lcs}), with probability at least $1-\delta$, we have i) $T =\{t_1, t_2, \ldots, t_m\}$ is compatible with all $\lambda \in S\setminus T$, and ii) $T_{\bar{i}} =\{t_1, \ldots, t_{i-1}, t_{i+1}, \ldots, t_m\}$ is compatible with $t_i$. By a union bound across all $i$ ($1\le i\le m\le k$), we know that with probability at least $1-k\delta$, the above two properties hold for all $i$. Let $\mathcal{E}$ denote such a random event, and then $\Pr[\mathcal{E}]\ge 1-k\delta$. It suffices to prove that conditioned upon $\mathcal{E}$, $X$ will output and only output the data items sent by its $m$ neighboring senders.
All our following discussions will condition on $\mathcal{E}$.

It suffices to prove that i) $X$ will output at most $m$ data items, and ii) for any neighboring sender $Y$ of $X$, $X$ will output the data items sent by $Y$. For the first part, note that conditioned upon $\mathcal{E}$, the multi-set $T$ is compatible with all $\lambda \in S\setminus T$. This means for all $\lambda \in  S\setminus T$, in Step 2 of the decoding part in Algorithm 1, we will have $\lambda\cdot z \le \sum_{i=1}^m (\lambda \cdot t_i) \le \frac{w}{2} < \frac{3w}{4}$, and hence Algorithm 1 will not output $\lambda$. Thus Algorithm 1 and 2 will output at most $m$ masking strings and $m$ data items, respectively.

We move on to prove the second part, and consider any neighboring sender $Y$ of $X$. Without loss of generality, assume $t_1$ is the masking string chosen by $Y$. In the decoding part of Algorithm~\ref{alg:masking}, in any slot where $t_1$ is $1$, $X$ will see either a ``1'' bit or the collision of multiple ``1'' bits. By the assumption on demodulation, the demodulation on $X$ will return a ``1'' bit for such a slot. Hence in the decoding part of Algorithm 1, we must have $t_1\cdot z = w > \frac{3w}{4}$, and Algorithm~\ref{alg:masking} must output $t_1$.
Next since the multi-set $\{t_2, t_3, \ldots, t_m\}$ is compatible with $t_1$, there will be at most $\frac{w}{2}$ possible $r$'s such that $t_1[r] = 1$ and $t_j[r] =1$ for some $j$ where $2\le j\le m$. Hence out of the total $w$ non-``blank'' RS symbols sent by $Y$, the receiver $X$ will obtain at least $w-\frac{w}{2} = \frac{w}{2}$ RS symbols at Step 5 in the decoding part of Algorithm 2. Since the RS coding rate was $\frac{1}{2}$, the RS decoding must succeed at Step 7, and the CRC checking must pass at Step 8. Hence Algorithm 2 will output the data item sent by $Y$.
\end{proof}

\Paragraph{Complexity of BMC encoding/decoding.}
We prove that the space and time complexities of Algorithm 1 and 2 are all low-order polynomials:
\begin{theorem}
\label{the:complexity}
The space complexity of Algorithm 1 and 2 combined is $O(\frac{kd}{\delta}\ln k)$. Let $\alpha$ ($\beta$) be the RS and CRC encoding (decoding) time complexity for one data item.
With probability of at least $1-k\delta$ and amortized for each data item, the encoding time complexity of Algorithm 1 and 2 combined is $O(kd+\alpha)$, and the decoding time complexity is $O(\frac{d}{\delta \ln d} + \frac{kd}{\ln d} + \beta)$.
\end{theorem}
\begin{proof}
The only non-trivial space complexity in Algorithm 1 and 2 is for storing the LCS $S$. $S$ has $\Theta(\frac{k}{\delta})$ masking strings, where each masking string takes $w\log_2 (4k) = \Theta(w\log k)$ bits to store.
Hence the total space complexity is $\Theta(\frac{k}{\delta}w\log k) = O(\frac{kd}{\delta}\log k)$.

The encoding time complexity is obvious. For decoding in Algorithm 1, we need to compute an inner product between $z$ and every $\lambda$ in $S$. To do so, we use the $w$ positions of the ``1'' bits in $\lambda$ to index into $z$. This will lead to $O(w) = O(\frac{d}{\ln d})$ (since $u = \Theta(\ln w) = \Theta(\ln d)$) complexity for each $\lambda$, or $O(\frac{kd}{\delta \ln d})$ for all $\lambda \in S$. The decoding in Algorithm 2 has total $|T|$ iterations. In each iteration, it constructs an $x$ while incurring $O(w|T|) = O(\frac{d}{\ln d} |T|)$  complexity, and then invokes RS and CRC decoding on $x$. By the proof of Theorem~\ref{the:main}, we know that with probability at least $1-k\delta$, $|T|\le k$. Hence the time complexity of Algorithm 1 and 2 combined will be $O(\frac{kd}{\delta \ln d} + k(\frac{kd}{\ln d} + \beta))$, or
$O(\frac{d}{\delta \ln d} + \frac{kd}{\ln d} + \beta)$ when amortized for each data item.
\end{proof}

\subsection{Practical Considerations}
\label{sec:practical}

\Paragraph{Errors during transmission.}
To facilitate understanding, so far we have not considered errors in transmission. Tolerating errors turns to be straightforward in Algorithm 1 and 2. First, Algorithm 1 actually already tolerates $\frac{w}{4}-1$ errors. The reason is that when there are no errors, for a masking string $\lambda$ that was sent by some sender, we have $\lambda\cdot z = w$. While for a $\lambda$ not sent by anyone, we will have $\lambda\cdot z \le \frac{w}{2}$. Hence, there is already a gap of $\frac{w}{2}$ for accommodating errors. Second, Algorithm 2 already uses RS coding internally. To tolerate errors in transmission, we can naturally add more redundancy in the RS code.

\Paragraph{Overhead of sending masking strings.}
In BMC, the senders need to first send masking strings before sending their data items. Such extra overhead turns out to be small: A sender sends $4kw$ RS symbols for the data item, and $4kw$ bits for the masking string. For 2-byte RS symbols, the overhead of sending the masking string is only $6.25$\% of that for the data item. Such overhead further decreases as RS symbol size increases. Furthermore, in practice, masking strings do not need to be re-sent for every data item. If the network topology never changes, then the masking strings only need to be sent once, and never need to be re-sent. Otherwise if the higher-level protocol is capable of detecting topology changes (e.g., when a sender newly moves into the communication range of a receiver), then the higher-level protocol can initiate/schedule the re-sending of masking strings in the network in response to such changes.




\section{Finding a Low Collision Set}
\label{sec:lcs}

BMC (more precisely, Theorem~\ref{the:main}) critically relies on the possibility of finding an LCS. This section will confirm that LCS indeed exists and can be found. Specifically, we will show that if we construct a multi-set in a certain randomized way, then with probability close to $1$, this multi-set will satisfy some {\em sufficient condition} for being an LCS and hence must be an LCS. We will further show that one can verify, in polynomial time, whether a multi-set satisfies this sufficient condition.
We remind the reader that the LCS is constructed prior to the deployment of BMC, and needs to be done only once.



\subsection{A Random Construction}
\label{sec:construct}

We use the following (simple) way of constructing a multi-set $S$ of $\frac{2k}{\delta}$ random masking strings, each of which is constructed independently.~\footnote{$S$ may contain duplicates, and hence it is a multi-set. Later we will prove that with good probability, $S$ actually has no duplicates.} To construct a random masking string with $4kw$ bits, for each $4k$-bit segment of the string, we set a uniformly random bit in the segment to be ``1'' and all remaining bits to be ``0''.

\subsection{Overview of Proof}

We want to show that with probability at least $0.95$, the multi-set returned by the above construction is an LCS. Despite the simplicity of the construction, the reasoning is rather complex because there are two random processes involved: The construction is random, while the definition of LCS (Definition~\ref{def:lcs}) also involves its own separate random process.

To decouple these two random processes, we will define another concept of {\em promising sets} (see Section~\ref{sec:promisingset}).
Different from LCS, the definition of a promising set involves only deterministic properties. Also as an important consequence, we will be able to verify, deterministically in polynomial time, whether a set is a promising set or not. In contrast, it is unclear how one can check (in polynomial time) whether a set is an LCS.
We will then later prove:
\begin{description}
\item[Claim 1.] With probability at least $0.95$, the multi-set returned by the random construction in Section~\ref{sec:construct} is a promising set (Theorem~\ref{the:promisingexist}).
\item[Claim 2.] A promising set must be an LCS (Theorem~\ref{the:promising2low}) --- namely, being a promising set is a {\em sufficient condition} for being an LCS.
\end{description}
We will only prove the above two claims for a certain given (small) $w$ value --- Theorem~\ref{the:promisingexist} and \ref{the:promising2low} only prove\footnote{Recall from Section~\ref{sec:overview} that we assume $d = \omega(\ln^2 N\times \ln\ln N)$. Section~\ref{sec:codingdata} further mentioned that $d = \omega(\ln^2 N\times \ln\ln N)$ implies $w = \omega(\ln^2 N)$. One can easily verify that as long as $\delta$ is not too small (e.g., as long as $\delta > \frac{1}{N^{5}}$), $w$ will be larger than $20(\ln\frac{k}{\delta})(\ln\frac{2k}{\delta^2})$ asymptotically.} for $w = 20(\ln\frac{k}{\delta})(\ln\frac{2k}{\delta^2})$.
This is because given an LCS for a small $w$ value,
we can easily get an LCS for larger $w$ values, by trivially extending each masking string:
\begin{theorem}
\label{the:extend}
Given any $(k, w, \delta)$ low collision set $S$ and any positive integer $c$, we can always construct a $(k, cw, \delta)$ low collision set $S^c$.
\end{theorem}
\begin{proof}
Let $S^c = \{\lambda^c \,\,|\,\, \lambda\in S \}$, where $\lambda^c$ refers to repeating $\lambda$ for $c$ times. For all $\lambda$ and $\eta$, we obviously have $\lambda^c \cdot \eta^c = c\times (\lambda \cdot \eta )$. Let $t_i^c$ ($1\le i\le m$) be any masking string from $S^c$. It is easy to verify that: i) $T^c = \{t_1^c, \ldots, t_m^c\}$ is compatible with all $\lambda^c \in S^c\setminus T^c$ iff $T = \{t_1, \ldots, t_m\}$ is compatible with all $\lambda \in S\setminus T$, and ii) for all $i$, $T_{\bar{i}}^c = \{t_1^c, \ldots, t_{i-1}^c,$ \mbox{} $t_{i+1}^c, \ldots, t_m^c\}$ is compatible with $t_i^c$ iff $T_{\bar{i}} = \{t_1, \ldots, t_{i-1}, t_{i+1}, \ldots, t_m\}$ is compatible with $t_i$.
A simple coupling argument will then show that since $S$ is an LCS, $S^c$ must be an LCS as well.
\end{proof}

\ifthenelse{\boolean{shortversion}}{

\subsection{Formal Results}
\label{sec:promisingset}

\Paragraph{The concept of promising sets.}
Recall the definition of inner product ($\cdot$) from Section~\ref{sec:lcsdef}.
Given a set $S$ of masking strings and any $\lambda\in S$, we define
$\mu(\lambda, S) = \frac{\sum_{s\in S \setminus \{\lambda\}} (\lambda\cdot s)}{|S|-1}$.
The following defines the concept of {\em promising sets}:
\begin{definition}
A set $S$ of $(k,w)$ masking strings is a {\em $(k,w, \delta)$ promising set} iff for all $\lambda\in S$, all the following equations hold:
\begin{eqnarray}
\label{eqn:proof1}
|\mu(\lambda, S) - \frac{w}{4k}| &<& \frac{0.04w}{4k} \\
\label{eqn:proof2}
\max_{s\in S \setminus \{\lambda\}} |\lambda\cdot s - \mu(\lambda, S)| &<& 4\ln\frac{k}{\delta} \\
\label{eqn:proof3}
\sum_{s\in S \setminus \{\lambda\}} (\lambda\cdot s - \mu(\lambda, S))^2
&<& (|S|-1) \frac{w}{5k} \ln\frac{k}{\delta}
\end{eqnarray}
\end{definition}

To get some intuition behind the above concept, note that $\mu(\lambda, S)$ is the average number of collisions between $\lambda$ and other masking strings in $S$. Equation~\ref{eqn:proof1} requires this average to be close to $\frac{w}{4k}$. Equation~\ref{eqn:proof2} requires the maximum number of collision to be close to this average.  Equation~\ref{eqn:proof3} bounds the ``variance'' of the number of collisions between $\lambda$ and other masking strings in $S$. The
values on the right-hand side of the equations
are carefully chosen such that i) the
random construction returns a promising set with good probability, and ii) a promising set must be an LCS.

\Paragraph{Formal theorems.}
The next two theorems show that i) with probability at least $0.95$, the multi-set returned by the random construction in Section~\ref{sec:construct} is a promising set, and ii) a promising set must be an LCS. For space constraints, the (lengthy) proofs of these two theorems are deferred to the full version~\cite{bondorf18} of this paper.

\begin{theorem}
\label{the:promisingexist}
Consider any $\delta$ where $0<\delta\le 0.02$, any $k$ where\footnote{The theorem requires $k\ge 6\ln\frac{2k}{\delta^2}$. Recall from Section~\ref{sec:overview} that we assume $k =  \omega(\ln N)$. One can easily verify that as long as $\delta$ is not too small (e.g., as long as $\delta > \frac{1}{N^{5}}$), $k$ will be larger than $6\ln\frac{2k}{\delta^2}$ asymptotically.} $k\ge 6\ln\frac{2k}{\delta^2}$, and $w = 20(\ln\frac{k}{\delta})(\ln\frac{2k}{\delta^2})$. With probability at least $0.95$, where the probability is taken over the random choices used in the construction, the multi-set $S$ constructed in Section~\ref{sec:construct} is a $(k, w, \delta)$ \linebreak promising set of size $\frac{2k}{\delta}$.
\end{theorem}

\begin{theorem}
\label{the:promising2low}
For all $0<\delta\le 0.02$, $k\ge 1$, and
$w = 20(\ln\frac{k}{\delta})(\ln\frac{2k}{\delta^2})$, a $(k, w, \delta)$  promising set $S$ of size $\frac{2k}{\delta}$ must be a $(k, w, \delta)$ LCS.
\end{theorem}

}{
\subsection{The Concept of Promising Set}
\label{sec:promisingset}

Recall the definition of inner product ($\cdot$) from Section~\ref{sec:lcsdef}.
Given a set $S$ of masking strings and any $\lambda\in S$, we define
$\mu(\lambda, S) = \frac{\sum_{s\in S \setminus \{\lambda\}} (\lambda\cdot s)}{|S|-1}$.
The following defines the concept of {\em promising sets}:
\begin{definition}
A set $S$ of $(k,w)$ masking strings is a {\em $(k,w, \delta)$ promising set} iff for all $\lambda\in S$, all the following equations hold:
\begin{eqnarray}
\label{eqn:proof1}
|\mu(\lambda, S) - \frac{w}{4k}| &<& \frac{0.04w}{4k} \\
\label{eqn:proof2}
\max_{s\in S \setminus \{\lambda\}} |\lambda\cdot s - \mu(\lambda, S)| &<& 4\ln\frac{k}{\delta} \\
\label{eqn:proof3}
\sum_{s\in S \setminus \{\lambda\}} (\lambda\cdot s - \mu(\lambda, S))^2
&<& (|S|-1) \frac{w}{5k} \ln\frac{k}{\delta}
\end{eqnarray}
\end{definition}

To get some intuition behind the above concept, note that $\mu(\lambda, S)$ is the average number of collisions between $\lambda$ and other masking strings in $S$. Equation~\ref{eqn:proof1} requires this average to be close to $\frac{w}{4k}$. Equation~\ref{eqn:proof2} requires the maximum number of collision to be close to this average.  Equation~\ref{eqn:proof3} bounds the ``variance'' of the number of collisions between $\lambda$ and other masking strings in $S$. The values on the right-hand side of the three equations are carefully chosen
such \mbox{}
that i) the random construction returns a promising set with good probability, and ii) a promising set must be an LCS.


\subsection{Probability of Being a Promising Set}
\label{sec:claim1}

The following proves that the probability of the random construction in Section~\ref{sec:construct} being a promising set.

\begin{theorem}
\label{the:promisingexist}
Consider any $\delta$ where $0<\delta\le 0.02$, any $k$ where\footnote{The theorem requires $k\ge 6\ln\frac{2k}{\delta^2}$. Recall from Section~\ref{sec:overview} that we assume $k =  \omega(\ln N)$. One can easily verify that as long as $\delta$ is not too small (e.g., as long as $\delta > \frac{1}{N^{5}}$), $k$ will be larger than $6\ln\frac{2k}{\delta^2}$ asymptotically.} $k\ge 6\ln\frac{2k}{\delta^2}$, and $w = 20(\ln\frac{k}{\delta})(\ln\frac{2k}{\delta^2})$. With probability at least $0.95$, where the probability is taken over the random choices used in the construction, the multi-set $S$ constructed in Section~\ref{sec:construct} is a $(k, w, \delta)$ \linebreak promising set of size $\frac{2k}{\delta}$.
\end{theorem}
\begin{proof}
Let $S = \{s_1, s_2, \ldots, s_\frac{2k}{\delta}\}$ be the multi-set constructed in Section~\ref{sec:construct}. With slight abuse of notation, for any $i$, we define
$\mu(s_i, S) = (\sum_{j, j\ne i} s_i\cdot s_j) / (|S|-1)$. We will later prove that, with probability at least $0.95$, the following holds for all $i$:
\begin{eqnarray}
\label{eqn:proof4}
|\mu(s_i, S) - \frac{w}{4k}| &<& \frac{0.04w}{4k} \\
\label{eqn:substitute}
\max_{j, j\ne i} |s_i \cdot s_j - \frac{w}{4k}| &<& 3.96\ln\frac{k}{\delta}\\
\label{eqn:substitute2}
\sum_{j, j\ne i}(s_i\cdot s_j-\frac{w}{4k})^2
&<& (|S|-1) \frac{w}{5k} \ln\frac{k}{\delta}
\end{eqnarray}
Note that Equation~\ref{eqn:substitute} implies $S$ being a set: If there existed $i$ and $j$ such that $i\ne j$ and $s_i = s_j$, then we would have $\max_{j, j\ne i} |s_i\cdot s_j - \frac{w}{4k}| = w - \frac{w}{4k} > 3.96\ln\frac{k}{\delta}$, violating Equation~\ref{eqn:substitute}.
Now given that $S$ is a set, Equation~\ref{eqn:proof4} becomes equivalent to Equation~\ref{eqn:proof1}.
Combining Equation~\ref{eqn:proof4} and Equation~\ref{eqn:substitute} will lead to Equation~\ref{eqn:proof2}, since\linebreak
$\max_{s\in S \setminus \{\lambda\}} |\lambda\cdot s - \mu(\lambda, S)| =
\max_{j, j\ne i} |s_i\cdot s_j - \mu(s_i, S)|$ \linebreak
$\le |\mu(s_i, S) - \frac{w}{4k}| + \max_{j, j\ne i} |s_i \cdot s_j - \frac{w}{4k}| \le \frac{0.04w}{4k} + 3.96\ln\frac{k}{\delta}$ \linebreak
$\le 4\ln\frac{k}{\delta}$.
Finally, note that $\mu(s_i, S)$ is the average across all $s_i\cdot s_j$ for $j \ne i$. Hence it is easy to verify that for any real value $a$, we have $\sum_{j, j\ne i}(s_i\cdot s_j -\mu(s_i, S))^2 \le \sum_{j, j\ne i}(s_i\cdot s_j-a)^2$.
Take $a = \frac{w}{4k}$, and we can immediately see that Equation~\ref{eqn:substitute2} implies Equation~\ref{eqn:proof3}. This will complete our proof of $S$ being a promising set.

We will next show that Equations~\ref{eqn:proof4}, \ref{eqn:substitute}, and \ref{eqn:substitute2} hold with probabilities of at least $0.99$, $0.98$, and $0.98$, respectively. A trivial union bound then shows that with probability at least $0.95$, they all hold.

First for Equation~\ref{eqn:proof4},
consider any fixed $i$ and fixed $s_i$, and view the remaining masking strings in $S$ as random variables (as a function of the random choices in the construction). The quantity $\sum_{j, j\ne i} s_i\cdot s_j$ follows a binomial distribution with parameters $(|S|-1)w$ and $\frac{1}{4k}$. By the Chernoff bound, we have
$\Pr[|\mu(s_i, S) - \frac{w}{4k}| \ge \frac{0.04w}{4k}] =
\Pr[|\sum_{j, j\ne i} s_i\cdot s_j - (|S|-1)\frac{w}{4k}| \ge (|S|-1)\frac{0.04w}{4k}] \le 2 exp(-\frac{1}{3} \cdot (0.04)^2 \cdot (|S|-1)\frac{w}{4k}) = 2 exp(- \frac{8}{3000}(\frac{2}{\delta}-\frac{1}{k}) (\ln\frac{2k}{\delta^2}) (\ln\frac{k}{\delta}))
< 2exp(-2.7\ln\frac{k}{\delta}) = 2\cdot (\frac{\delta}{k})^{2.7}$.
By a union bound across all $\frac{2k}{\delta}$ possible $i$'s, we know that with probability at least $1- 4\cdot (\frac{\delta}{k})^{1.7} > 0.99$, Equation~\ref{eqn:proof4} holds.

Next for Equation~\ref{eqn:substitute}, consider any fixed $i$ and fixed $s_i$, and view $s_j$ as a random variable.
The quantity $s_i \cdot s_j$ follows a binomial distribution with parameters of $w$ and $\frac{1}{4k}$, and a mean of $\frac{w}{4k}$. Also note that since $k\ge 6\ln\frac{2k}{\delta^2}$, we have $3.96\ln\frac{k}{\delta} \ge 4.752\times \frac{w}{4k}$. By the Chernoff bound, we have:
\begin{eqnarray}
\nonumber
\Pr\Big[|s_i \cdot s_j - \frac{w}{4k}| \ge 3.96\ln\frac{k}{\delta}\Big]
&<& \bigg(\frac{e}{1+\frac{3.96\ln\frac{k}{\delta}}{\frac{w}{4k}}}\bigg)^
{\big(1+\frac{3.96\ln\frac{k}{\delta}}{\frac{w}{4k}}\big) \cdot \frac{w}{4k}} \\
\label{eqn:proof7}
&\hspace*{-2cm}<& \hspace*{-1cm}\Big(\frac{e}{1+4.752}\Big)^ {3.96\ln\frac{k}{\delta}} < \Big(\frac{\delta}{k}\Big)^{2.96}
\end{eqnarray}
There are total $|S|(|S|-1) < (\frac{2k}{\delta})^2$ possible combinations of $i$ and $j$. Take a union bound across all of these, we know that with probability at least $1- 4(\frac{\delta}{k})^{0.96} > 0.98$, Equation~\ref{eqn:substitute} holds for all $i$.

Finally for Equation~\ref{eqn:substitute2}, consider any fixed $i$ and fixed $s_i$, and view the remaining masking strings in $S$ as random variables. Under the given $i$ and $s_i$, define the random variable $X_j=
\frac{(s_i\cdot s_j-\frac{w}{4k})^2}{16\ln^2 \frac{k}{\delta}}$ for $j\ne i$. The quantity $s_i\cdot s_j$ is a binomial random variable with parameters $w$ and $\frac{1}{4k}$. Hence we have
$E[X_j] = \frac{E[(s_i\cdot s_j-\frac{w}{4k})^2]}{16\ln^2 \frac{k}{\delta}}$
$ = \frac{E[(s_i\cdot s_j-E[s_i\cdot s_j])^2]}
{16\ln^2 \frac{k}{\delta}} = \frac{\mbox{Var}[s_i\cdot s_j]}
{16\ln^2 \frac{k}{\delta}} = \frac{w\cdot \frac{1}{4k} \cdot (1-\frac{1}{4k})}{16\ln^2 \frac{k}{\delta}}
< \frac{w}{4k} \cdot \frac{1}{16\ln^2 \frac{k}{\delta}} <
\frac{w}{10k}(\ln\frac{k}{\delta}) \cdot \frac{1}{16\ln^2 \frac{k}{\delta}}$, and also
$E[\sum_{j,j\ne i}X_j] = (|S|-1)E[X_1] <\frac{w}{10k}(\ln\frac{k}{\delta}) \cdot \frac{|S|-1}{16\ln^2 \frac{k}{\delta}}$.
For the given $i$ and $s_i$, by Equation~\ref{eqn:proof7} and a union bound across all $j$, we know that with probability at least $1- 2(\frac{\delta}{k})^{1.96}$, $|s_i\cdot s_j -\frac{w}{4k}|< 3.96\ln\frac{k}{\delta}$ and hence $X_j < 1$. Conditioned upon such an event, we invoke the Chernoff bound and get \linebreak $\Pr\Big[\sum_{j, j\ne i}(s_i\cdot s_j-\frac{w}{4k})^2 \ge (|S|-1) \frac{w}{5k} \ln\frac{k}{\delta}\Big]
= \Pr\Big[\sum_{j, j\ne i} X_j \ge 2\cdot
\frac{w}{10k}(\ln\frac{k}{\delta}) \cdot \frac{|S|-1}{16\ln^2 \frac{k}{\delta}}\Big]
\le exp\Big(-\frac{1}{3} \frac{w}{10k}(\ln\frac{k}{\delta})
\cdot \frac{|S|-1}{16\ln^2 \frac{k}{\delta}}\Big) = exp\Big(-\frac{1}{3}
\frac{20(\ln\frac{k}{\delta})(\ln\frac{2k}{\delta^2})}{10k}
(\ln\frac{k}{\delta}) \cdot \frac{\frac{2k}{\delta}-1}{16\ln^2 \frac{k}{\delta}}\Big)
< exp\Big(-\frac{1}{24}(\frac{2}{\delta} - \frac{1}{k})\ln\frac{k}{\delta}\Big)
< \Big(\frac{\delta}{k}\Big)^{4.1}$.

Hence we know that for the given $i$, with probability at least $1-2(\frac{\delta}{k})^{1.96} - (\frac{\delta}{k})^{4.1}$, Equation~\ref{eqn:substitute2} hold.
Finally, take a union bound across all possible $i$'s, we know that with probability at least $1-4(\frac{\delta}{k})^{0.96} - 2(\frac{\delta}{k})^{3.1} > 0.98$, Equation~\ref{eqn:substitute2} holds for all $i$.
\end{proof}

\subsection{A Promising Set Must Be an LCS}
\label{sec:claim2}
\mbox{}\vspace*{-4mm}

\begin{theorem}
\label{the:promising2low}
For all $0<\delta\le 0.02$, all $k\ge 1$, and
$w = 20(\ln\frac{k}{\delta})(\ln\frac{2k}{\delta^2})$, a $(k, w, \delta)$  promising set $S$ of size $\frac{2k}{\delta}$ must be a $(k, w, \delta)$ LCS.
\end{theorem}
\begin{proof}
Imagine that we choose $m\le k$ elements, $t_1$ through $t_m$, from $S$ uniformly randomly with replacement, and define the multi-set $T = \{t_1, t_2, \ldots, t_m\}$.
Remember that here $S$ is already fixed --- only the $t_i$'s and $T$ are random variables.

We will first prove that $S$ satisfies the first requirement of LCS. Specifically, we will show that with probability at least $1-0.44\delta$, $\sum_{i=1}^m \lambda\cdot t_i \le \frac{w}{2}$ for all $\lambda\in S\setminus T$. We consider a binary matrix where its $|S|^k$ columns correspond to all the possible $T$'s, and its $|S|$ rows correspond to all the possible $\lambda\in S$. A matrix entry corresponding to given $T$ and $\lambda$ is {\em bad} iff $\sum_{i=1}^m \lambda\cdot t_i > \frac{w}{2}$ and $\lambda\in S\setminus T$. To prove the earlier claim, it suffices to show that at least $|S|^k\times (1-0.44\delta)$ columns contain no bad entries. Directly proving this will be tricky, so we instead prove that for each row, at most $\frac{0.44\delta}{|S|}$ fraction of the entries are bad. This will then imply that the total number of bad entries in the matrix is at most $|S|^k \times |S| \times \frac{0.44\delta}{|S|} = |S|^k \times 0.44\delta$, and hence there can be at most $|S|^k \times 0.44\delta$ columns containing bad entries.

Now to prove that each row has at most $\frac{0.44\delta}{|S|}$ fraction of its entries being bad, it suffices to prove that for any given $\lambda$, when we choose $t_1$ through $t_m$ from $S\setminus \{\lambda\}$ uniformly randomly with replacement, we will have:
\begin{eqnarray}
\label{eqn:proof8}
\Pr\Big[\sum_{i=1}^m \lambda\cdot t_i \ge \frac{w}{2}\Big] &\le& \frac{0.44\delta}{|S|}
\end{eqnarray}
To prove Equation~\ref{eqn:proof8}, define $Y_i= \lambda\cdot t_i-\mu(\lambda, S)$ for $1\le i \le m$, and we directly have $\mbox{E}[Y_i]=0$. (Note that $\lambda\cdot t_i$ does {\em not} follow a binomial distribution.)
Since $S$ is a promising set, by Equation~\ref{eqn:proof1} we have
$\sum_{i=1}^{m}\lambda\cdot t_i = \sum_{i=1}^{m}(Y_i+\mu(\lambda, S)) = m\cdot \mu(\lambda, S) + \sum_{i=1}^{m} Y_i
< \frac{1.04w}{4} + \sum_{i=1}^{m} Y_i$.
Next, for all $i$, Equation~\ref{eqn:proof2} and \ref{eqn:proof3} tell us that
$|Y_i| < 4\ln\frac{k}{\delta}$ and $\mbox{E}[Y_i^2]< \frac{w}{5k} \ln\frac{k}{\delta}$.
By Bernstein's inequality~\cite{uspensky37}, we have
$\Pr\Big[\sum_{i=1}^m Y_i>\frac{0.96w}{4}\Big]
\le \exp\bigg(-\frac{\frac{(0.96w)^2}{16}}{2m\cdot \frac{w}{5k} \ln\frac{k}{\delta}+ \frac{2}{3} \cdot 4\ln\frac{k}{\delta} \cdot \frac{0.96w}{4}}\bigg)
\le exp\Big(-\frac{\frac{0.9216}{16}w}{1.04\ln\frac{k}{\delta}}\Big) = exp\Big(-\frac{18.432}{16.64} \ln\frac{2k}{\delta^2}\Big) < \Big(\frac{\delta^2}{2k}\Big)^{1.107} = $\linebreak $\Big(\frac{\delta^2}{2k}\Big)^{0.107}\cdot \frac{\delta}{|S|} < \frac{0.44\delta}{|S|}$.
In turn, for any given $\lambda\in S$, Equation~\ref{eqn:proof8} follows since
$\Pr[\sum_{i=1}^m \lambda\cdot t_i \ge \frac{w}{2}]
\le \Pr[\sum_{i=1}^m Y_i>\frac{0.96w}{4}]
\le \frac{0.44\delta}{|S|}$.

We next prove that $S$ satisfies the second requirement for an LCS. Specifically, we will show that for any fixed $i$ where $1\le i\le m$, with probability at least $1-0.505\delta$, the multi-set $\{t_1, \ldots, t_{i-1}, t_{i+1}, \ldots, t_m\}$ is compatible with $t_i$. We obviously only need to prove this for $m\ge 2$. Since all the $t_i$'s are symmetric, without loss of generality, assume $i = m$. Define $T_{\overline{m}} = \{t_1, t_2, \ldots, t_{m-1}\}$. We claim that with probability at least $1-0.5\delta$, $t_m\notin T_{\overline{m}}$. To see why, note that $t_1$ through $t_{m-1}$ corresponds to at most $m-1$ distinct elements form $S$, and hence
$\Pr[t_m\in T_{\overline{m}}] \le \frac{k-1}{|S|} < 0.5\delta$.
Now conditioned upon $t_m\notin T_{\overline{m}}$, each $t_j$ for $1\le j\le m-1$ is a uniformly random string in $S\setminus \{t_m\}$. One can now apply a similar proof as for Equation~\ref{eqn:proof8} (after replacing $m$ with $m-1$), and show that
$\Pr[\sum_{j=1}^{m-1} (t_m \cdot t_j) \ge \frac{w}{2}]
\le \frac{0.44\delta}{|S|} \le 0.0044\delta$. Hence we know that with probability at least $(1-0.5\delta)\cdot (1-0.0044\delta) > 1-0.5044\delta$, the multi-set $T_{\overline{m}}$ is compatible with $t_m$.

Finally, a trivial union bound across the two requirements shows that $S$ is an LCS.
\end{proof}

}

\section{Numerical Examples}
\label{sec:exp}

To supplement the formal guarantees of BMC, this section presents some basic numerical examples.
%
%
In the context of our example scenario (Figure~\ref{fig:app}), we consider a receiver with $k=100$ neighboring senders. Each sender has a data item (including CRC) of $d$ bytes, to be sent to the receiver. We will consider $d = 25$ to $100$. We will use RS symbol size of $1$ byte, and hence $w = 50$ to $200$. BMC requires an LCS $S$. Our experiments will directly use the multi-set constructed in Section~\ref{sec:construct} as $S$, with $|S|= 2\times 10^6$.

\subsection{Overhead of BMC Encoding/Decoding}

\Paragraph{Space overhead.}
Recall that each masking string takes $w \log_2(4k)$ bits to store. Storing $S$ thus takes no more than $500$MB under our previous parameters.
%
%
On a sender, it is possible to further reduce such overhead.
Recall that a sender only needs to pick a random masking string from $S$. In practice, the sender may just pick a random masking string beforehand, and store that masking string (incurring only about $60$ to $250$ bytes).
Such asymmetric overhead is a salient feature of BMC: For example, the senders may be resource-constrained sensors, while the receivers may be more powerful.

\Paragraph{Computation overhead.}
For BMC encoding, Algorithm~\ref{alg:masking} and \ref{alg:data} show that the computational overhead mostly comes from RS encoding. Since the overhead of RS code is well understood~\cite{lin04ecc,taipale94}, we do not provide separate results here due to space constraints.

BMC decoding has two phases, for decoding masking strings and data items, respectively. Decoding masking strings entails an exhaustive enumeration of all masking strings in $S$. We have implemented the masking string decoding algorithm as a single-threaded Java program. We observe that under our previous parameters, on average it takes about $1.06$ms to $4.72$ms to decode one masking string,
on a 3.4GHz desktop PC. Following the discussion in Section~\ref{sec:practical}, in practice, masking strings will only need to be re-sent and re-decoded when the network topology changes. Hence
such decoding cost can be easily amortized across many (e.g., $100$) data items. Also note that our decoding algorithm can be easily made parallel, and thus will likely enjoy a linear speedup when running over multiple cores.

The computational overhead in the second phase of BMC decoding is dominated by RS decoding. Again, we do not provide separate evaluation here since such overhead is well understood~\cite{lin04ecc,taipale94}.

\subsection{BMC vs. Baselines}

We consider a scenario with $t = 100,000$ bytes of airtime available for the $k$ senders to send their respective data items to the receiver. Here, one {\em byte of airtime} is the transmission airtime of exactly one byte. The byte airtime needed by BMC is $9kd$, which ranges from $22,500$ to $90,000$ under our parameters, and is always below $t$. We will use {\em failure rate} as the measure of goodness, defined as the fraction of data items that are not successfully delivered by the $t$ deadline.


\Paragraph{Four schemes to compare.}
We consider two baselines. The first baseline {\tt RandAccess1} divides the available total time into $l = \frac{t}{d}$ {\em intervals}. In each interval, independently with probability $\frac{1}{k}$, a sender sends its data item. One can easily verify that the probability of $\frac{1}{k}$ maximizes the utilization of the channel. A data item is considered delivered {\em successfully} if there exists at least one interval during which the corresponding sender is the sole sender. The second baseline {\tt RandAccess2} is the same, except that each sender chooses exactly $\frac{l}{k}$ distinct intervals out of the $l$ intervals, in a uniformly random fashion.

We also consider two versions of BMC. The first version {\tt BMC1} simulates the behavior of Algorithm 1 and 2.
We assume that CRC has no false negatives, and hence do not simulate it explicitly. We do not simulate RS code encoding/decoding either --- instead, since we use a coding rate of $\frac{1}{2}$ in the RS code, our simulation will assume that RS decoding succeeds iff the number of erased RS symbols is at most $\frac{w}{2}$.
Note that the byte airtime needed by BMC is $9kd$, and hence {\tt BMC1} may not fully utilize the available time $t$. In {\tt BMC2}, each sender will repeat its actions in {\tt BMC1}, for $\lfloor\frac{t}{9kd}\rfloor$ times.

Because the failure rate under {\tt BMC2} can be rather small, it may take excessive simulation time to observe any failure. Hence the results under {\tt BMC2}  are directly {\em computed} from those under {\tt BMC1}, rather than from simulation. For example, if the failure rate under {\tt BMC1} is $0.1$ and if $\lfloor\frac{t}{9kd}\rfloor = 3$, we will plot $0.001$ as the failure rate under {\tt BMC2}. The failure rates for all other schemes are directly obtained from simulation.
When the failure rate is small, we increase the number of trials to observe a sufficient number of failures.



\begin{figure}
	\centering
	\includegraphics[width=\columnwidth]{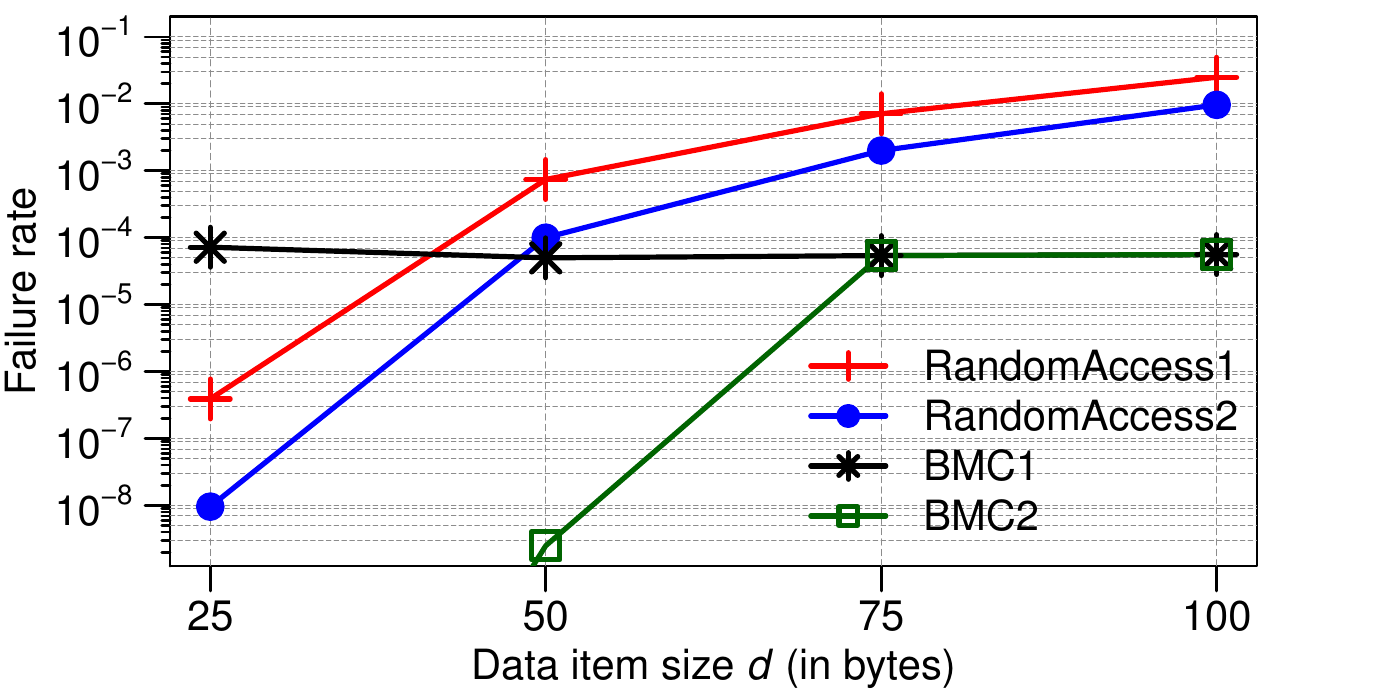}
	\vspace{-5mm}
	\caption{Failure rate of different schemes.}
	\vspace{-7mm}
	\label{fig:resultdelivery}
\vspace*{2mm}
\end{figure}


\Paragraph{Comparison.}
Figure~\ref{fig:resultdelivery} compares the failure rates of the four schemes. {\tt BMC2} consistently achieves a failure rate about 2 orders of magnitude smaller than the two baselines. {\tt BMC1} does not always do so because it does not actually use up the available airtime: Under $d = 25$, {\tt BMC1} uses only $22$\% of the airtime available.
The failure rate of {\tt BMC1} is largely independent of the data item size. This is expected since its failure rate is largely determined by the size of $S$.
For all other schemes, the failure rate increases with the data item size, since under the given time constraint and with larger data items, they have fewer opportunities to send.


\section{Conclusions}
\label{sec:conclusions}

Given the fundamental limit of $\mathbf{R} = O(\frac{1}{\ln N})$ for
scheduling in multi-sender multi-receiver wireless networks~\cite{ghaffari12}, our ultimate goal is to achieve $\mathbf{R} = \Theta(1)$ and also to avoid all the complexities in scheduling. As the theoretical underpinning for achieving our ultimate goal, this work proposes BMC and proves that BMC enables wireless networks to achieve $\mathbf{R} = \Theta(1)$. We hope that our theoretical results can attest the promise of this direction, and spur future systems research (especially on the physical layer) along this line.

\begin{acks}
We thank our anonymous shepherd and the anonymous IPSN reviewers for their helpful feedbacks, which significantly improved this paper. We thank Rui Zhang for helpful discussions, and Sidharth Jaggi for helpful comments regarding the sublinear-time group testing literature.
This work is partly supported by the research grant MOE2017-T2-2-031 from Singapore Ministry of Education Academic Research Fund Tier-2. Binbin Chen is supported by the National Research Foundation, Prime Minister's Office, Singapore, partly under the Energy Programme administrated by the Energy Market Authority (EP Award No. NRF2017EWT-EP003-047) and partly under the Campus for Research Excellence
and Technological Enterprise (CREATE) programme. Jonathan Scarlett is supported by an NUS Early Career Research Award.
\end{acks}

\bibliographystyle{ACM-Reference-Format}
\bibliography{refs}

\end{document}